%% file: main.tex
\documentclass[10pt,lettersize,journal]{IEEEtran}
\usepackage{latexsym,amssymb,amsmath,epsfig,subfigure,algorithm,mathtools}

\usepackage{subcaption}

\usepackage{amsmath,amsfonts}
\usepackage{algorithmic}
\usepackage{array, footnote, cite}
\usepackage[caption=false,font=normalsize,labelfont=sf,textfont=sf]{subfig}
\usepackage{textcomp}
\usepackage{stfloats}
\usepackage{url}
\usepackage{verbatim}
\usepackage{graphicx}
\def\BibTeX{{\rm B\kern-.05em{\sc i\kern-.025em b}\kern-.08em
    T\kern-.1667em\lower.7ex\hbox{E}\kern-.125emX}}
\usepackage{balance}
\usepackage{multirow}
\usepackage{colortbl}
\usepackage{soul}

\AtBeginDocument{
    \setlength{\abovedisplayskip}{3pt}
    \setlength{\belowdisplayskip}{3pt}
    \setlength{\abovedisplayshortskip}{3pt}
    \setlength{\belowdisplayshortskip}{3pt}
}
\usepackage{amsthm}
\usepackage{thmtools, thm-restate}
\usepackage{algorithm}
\input{symbols_commands}
\newcommand{\alg}{$\mathsf{ROAR-Fed}~$}

\newcommand{\algns}{$\mathsf{ROAR-Fed}$}
\newcommand{\algp}{$\mathsf{PROAR-PFed}$}
\newcommand{\alghota}{$\mathsf{HOTA-FedGradNorm}$}

\title{Adaptive Personalized Over-the-Air Federated Learning with Reflecting Intelligent Surfaces
\thanks{This work was supported in part by NSF CNS-2112471.}
\thanks{This paper was presented in part at IEEE International Conference on Communications (ICC), May 2023~\cite{mao22roar}, IEEE International Conference on Communications Workshops (ICC Workshop), May 2023~\cite{mao22iccw} and IEEE International Conference on Acoustics, Speech, and Signal Processing (ICASSP), April 2024~\cite{mao23personal}.}
\thanks{The authors are with the Department of Electrical and Computer Engineering, The Ohio State University, Columbus, OH 43210 USA (e-mail: mao.518@osu.edu; yener@ece.osu.edu).}}

\author{
  \IEEEauthorblockN{Jiayu Mao, {\it  Student Member, IEEE} and Aylin Yener, \it{Fellow, IEEE}}\\
  
}

\begin{document}

\sloppy
\allowdisplaybreaks[1]

\maketitle

\vspace{-0.5in}


\begin{abstract}
Over-the-air federated learning (OTA-FL) unifies communication and model aggregation by leveraging the inherent superposition property of the wireless medium. This strategy can enable scalable and bandwidth-efficient learning via simultaneous transmission of model updates using the same frequency resources, if care is exercised to design the physical layer jointly with learning. In this paper, a federated learning system facilitated by a heterogeneous edge-intelligent network is considered. The edge users (clients) have differing user resources and non-i.i.d. local dataset distributions. A general non-convex learning objective is considered for the model training task(s) at hand. 
We augment the network with Reconfigurable Intelligent Surfaces (RIS) in order to enhance the learning system.
We propose a cross-layer algorithm that jointly assigns communication, computation and learning resources. 
In particular, we adaptively adjust the number of local steps in conjunction with RIS configuration to boost the learning performance.
Our system model considers channel noise and channel estimation errors in both the uplink (model updates) and downlink (global model broadcast), employing dynamic power control for both.
We provide the convergence analysis for the proposed algorithms and extend the frameworks to personalized learning.
Our experimental results demonstrate that the proposed algorithms outperform the state-of-the-art joint communication and learning baselines.
\end{abstract}



\section{Introduction}
\label{sec:intro}
Federated learning (FL) \cite{mcmahan2017} has received considerable attraction and found a wide range of application in recent years as a promising distributed learning paradigm.
FL employs iterative local training coordinated by a parameter server (PS) for a potentially large number of users to collaboratively train a machine learning model without sharing their raw data.
Throughout each iteration, users train their local models using their own datasets and transmit only model parameters or updates to the PS, and the PS then aggregates these local updates to compute the global model for the next round.
More recently, personalized federated learning (PFL) has emerged for heterogeneous FL settings, i.e., non-i.i.d. datasets~\cite{li2021ditto}.
Although FL is a naturally suitable framework for mobile edge networks, care must be exercised when implementing it in wireless networks.
Taking into account the inherent features of wireless channels more efficient learning by the edge can be facilitated.

Over-the-air federated learning (OTA-FL) leverages the wireless medium's superposition property for model aggregation. 
All users simultaneously transmit their suitably scaled local updates using the same time-frequency resources, enabling the PS to directly obtain the analog aggregated model.
The majority of previous research on OTA-FL have mainly considered noise on uplink transmissions, and an error-free downlink \cite{amiri2020,zhu2019broadband}.
In~\cite{yang22}, we have provided an integrated computation and power control design to enhance the learning performance for heterogeneous systems. 
Recently, several works have started to focus on the impact of simultaneous noisy downlink and uplink~\cite{wei2022federated,zhang2023deep}. 
While these studies offer valuable insights into the convergence and energy efficiency challenges posed by noisy communication channels, they typically assume perfect channel state information (CSI), which is essential for precise over-the-air aggregation. In reality, only estimated CSI is available, leading to potential signal misalignment and degraded learning performance. Several efforts have been made to address this limitation.
In~\cite{shao2021federated}, the authors have proposed a whitened matched filtering and sampling scheme to address the issue of imperfect CSI on the receiver end.
In~\cite{evgenidis2023over}, the authors have introduced an adaptive pilot retransmission policy and optimization framework to minimize MSE under imperfect CSI.
In~\cite{mao22}, we have extended~\cite{yang22} to analyze the impact of imperfect CSI in convergence analysis.

Recently, reconfigurable intelligent surfaces (RIS), including passive RIS~\cite{wu2019intel}, active RIS~\cite{shi2024federated}, hybrid RIS~\cite{jin2023hybrid} and STAR-RIS~\cite{ni2022star}, have emerged as a promising solution to enable spectral efficiency and high reliability in next-generation networks.
Among them, passive RIS is a solution that consists of a flat meta-surface with numerous reflecting elements, which can be dynamically controlled to tune phase shifts in the desired manner.
Despite the recent advances in active and hybrid RIS types, passive RIS remains attractive with its cost-effectiveness and energy efficiency, making it ideal for large-scale and low-cost deployments in practical communication systems.
In this work, unless otherwise specified, the term RIS refers to passive RIS.
By reflecting the signals toward intended directions, RIS can proactively transform the radio propagation environment into a more favorable one, with judicious deployment.
As such, RIS has been leveraged to improve energy efficiency~\cite{zhang2021energy}, privacy~\cite{shi2024federated}, latency~\cite{le2023federated}, and rate~\cite{ni2022integrating}.
Given its advantages, RIS holds strong potential for integration with federated edge learning, enhancing model aggregation and overall learning performance~\cite{yang2020fed}.
Several recent works have explored RIS-assisted OTA-FL communication systems.
Reference~\cite{zheng2022balancing} aims to minimize the mean-squared error (MSE) of the aggregated model and retain the recoverability of local models by concurrently optimizing the RIS configuration and beamformers, considering scenarios with both perfect and imperfect CSI.
\cite{wang2021fed,ni2021fed,kim2024reconfigurable} maximize the number of selected devices under predefined MSE constraints by optimizing the device selection and RIS design, showing the convergence improvement brought by deploying RIS.
Reference~\cite{liu2021csit} leverages RIS to reduce downlink feedback overhead, configuring RIS phase shifts to align local model updates coherently when the CSI at the transmitter (CSIT) is unavailable.
\cite{yang2022federated} aims to optimize FL system utility as well as efficient spectrum learning with the aid of multiple RISs.
In~\cite{li2022one}, RIS is utilized to aid one-bit communication for sign-based FL system.
In~\cite{wang2023graph}, a graph neural network (GNN) based approach is developed to minimize the time-average error of convergence upper bound.
In~\cite{zhang2024irs}, a weight-selection based framework is introduced to minimize MSE in RIS-assisted OTA-FL system.
Only a few efforts have been made to explore unified communication and learning approaches for the RIS-assisted OTA-FL model.
Reference \cite{liu2021risfl} unifies communication and learning optimization by jointly designing beamforming, RIS configuration, and device selection while assuming a static time-invariant channel with perfect CSI.
In~\cite{zhao2023performance}, the problem of the optimality gap and the
energy consumption minimization is solved by employing a Lyapunov optimization framework.
In~\cite{zheng2024novel}, a RIS-assisted simultaneous wireless information and power transfer (SWIPT) and OTA-FL system is introduced, optimizing device selection, beamforming, and RIS design to improve convergence and accuracy.

Though having a significant interest in machine learning, the specific application of personalized FL (PFL) in wireless systems, especially in 6G environments, remains underexplored. 
Existing research, such as~\cite{sami2022over} on clustered federated learning and~\cite{mortaheb2022personalized, chen2023personalizing} on multi-task learning in wireless communication, overlooks the integration of RIS.
While RIS is considered in~\cite{shi2024empowering}, its focus remains on clustered FL.
By contrast, personal RIS devices offer a more tailored approach, aligning closely with individual client needs and enhancing the learning process.
The notion of personalized RIS in personalized learning has been first proposed in our conference paper~\cite{mao23personal}, building on single RIS-assisted OTA-FL framework in conference papers~\cite{mao22roar,mao22iccw} to develop a personal RIS-assisted framework where each user also needs their own personalized model.

Distinct from the majority of existing research that focuses on minimizing MSE, in this paper, we propose a cross-layer algorithm that concurrently optimizes communication and computation resources to improve global and personalized learning performance in RIS-assisted OTA-FL systems.
Unlike existing joint communication and learning design, we take into account the practical time-varying physical layer under imperfect CSI at edge devices.
We consider a general system model, which addresses a noisy uplink {\it and} downlink communication system. 
The presence of noise and imperfect CSI in the downlink introduces a noisy local starting point during each global iteration, which inevitably generates extra error terms in convergence analysis, an aspect that our previous work has not addressed.
We summarize our main contributions as follows.
\begin{list}{\labelitemi}{\leftmargin=1em \itemindent=-0.5em \itemsep=.2em}
\item We develop a joint communication and learning design that adaptively adjusts the RIS phase shifts (channel), the number of local update steps (computation) and transmission power (communication) in concert during each global iteration to alleviate the effects of both time-varying imperfect CSI and data/system heterogeneity.
Our proposed algorithm,~\alg (\underline{R}IS-assisted \underline{O}ver-the-air \underline{A}daptive \underline{R}esource Allocation for \underline{Fed}erated learning), employs dynamic power control for both downlink and uplink transmissions.
\item We provide the convergence analysis of the proposed algorithms for a general FL with non-convex objective, revealing the overall impacts of imperfect CSI and communication noise.
The convergence upper bound quantifies learning errors and communication errors, where some error terms contain both types of errors due to the cross-layer design.
\item We evaluate the performance of our algorithms in a heterogeneous network with non-i.i.d. local datasets and diverse device resources on two different datasets (MNIST, Fashion-MNIST).
Numerical results demonstrate that our algorithm significantly outperforms all state-of-the-art.
\item We extend the proposed framework to personalized FL with personal RISs, the extended algorithm is termed~\algp~(\underline{P}ersonal \underline{R}IS-assisted \underline{O}ver-the-\underline{A}ir \underline{R}esource Allocation for \underline{P}ersonalized \underline{Fed}erated Learning). We validate the potential of the personal RIS model for personalization through numerical experiments.
\end{list}

The remainder of the paper is organized as follows. 
Section~\ref{sec: prelim} presents the system model. 
We discuss our cross-layer approach in Section~\ref{sec: alg}.
The convergence analysis is provided in Section~\ref{sec: conv}. 
Section~\ref{sec: extension} introduces the extension to personalization.
In Section~\ref{sec: exp}, we show the numerical results. Section~\ref{sec: conclusion} concludes the paper. Proofs are provided in Appendix A-C.

\section{System Model} \label{sec: prelim}
\subsection{Federated Learning Model} 
\label{subsec: fl}
We consider a FL system consisting of a parameter server (PS) and $m$ edge devices.
Each device $i \in [m]$ contains its local dataset $D_i$ whose data points are randomly selected from the distribution $\mc{X}_i$. 
We assume that the distributions of local datasets are non-i.i.d., as is the case in practice. That is, $\mc{X}_i \neq \mc{X}_j$ if $ i \neq j, \forall i, j \in [m]$.
The goal of FL is to minimize a global empirical loss function:
\begin{equation}
    \min_{\w \in \mathbb{R}^q}F(\w) \triangleq \min_{\w\in\mathbb{R}^q} \sum_{i \in [m]} \alpha_i F_i(\w, D_i), 
    \label{eq: objective}
\end{equation}
where $\w$ is a $q$-dimension parameter vector of global model,  $\alpha_i = \frac{| D_i |}{\sum_{i \in [m]} | D_i |}$ represents the weight of edge device $i$, $F_i(\w, D_i) \triangleq \frac{1}{| D_i |} \sum_{\xi^i_j \in D_i} F(\w, \xi^i_j)$ is the local learning objective and $\xi^i_j$ is the $j$-th sample in $D_i$.
We assume that $F_i(\w, D_i)$ are general non-convex objective functions, and edge devices have datasets of different sizes, $\alpha_i \ne \alpha_j$, $i\ne j$.

In each round of FL, PS broadcasts the current global model parameters $\w_{t}$ to edge devices.
Each device $i$ starts local training after receiving the global model $\w_t^i$. Specifically, it performs stochastic gradient descent (SGD):
\begin{equation}
\w^i_{t, k+1} = \w^i_{t, k} - \eta_t \nabla F_{i}(\w^i_{t, k}, \xi^i_{t, k}), \quad k = 0,\ldots,\tau_t^i-1, \label{equ:sgd}
\end{equation}
where $k$ is the index of local step, $\xi^i_{t, k}$ is the randomly chosen data sample, $\eta_t$ is the local learning rate and $\tau_t^i$ represents the number of local steps. We design $\tau_t^i$ to vary across edge devices and each round.
Next, all edge devices upload their local updates to the PS. Then the PS aggregates and updates the global model parameter accordingly:
\begin{equation}
    \w_{t+1} = \w_t + \sum_{i \in [m]}\alpha_i(\w_{t,\tau_t^i-1}^i - \w_{t,0}^i).
\end{equation}
The above training procedure continues until convergence.
Note that in OTA-FL, communication and aggregation occur simultaneously at the PS as a result of the inherent superposition property of the wireless channel.

\vspace{-0.1in}
\subsection{RIS-Assisted Communication Model} 
\label{subsec: comm}

\subsubsection{Uplink Communication Model} 
\label{subsubsec: uplink}

\begin{figure}[t]
  \centering
  \includegraphics[width=\linewidth]{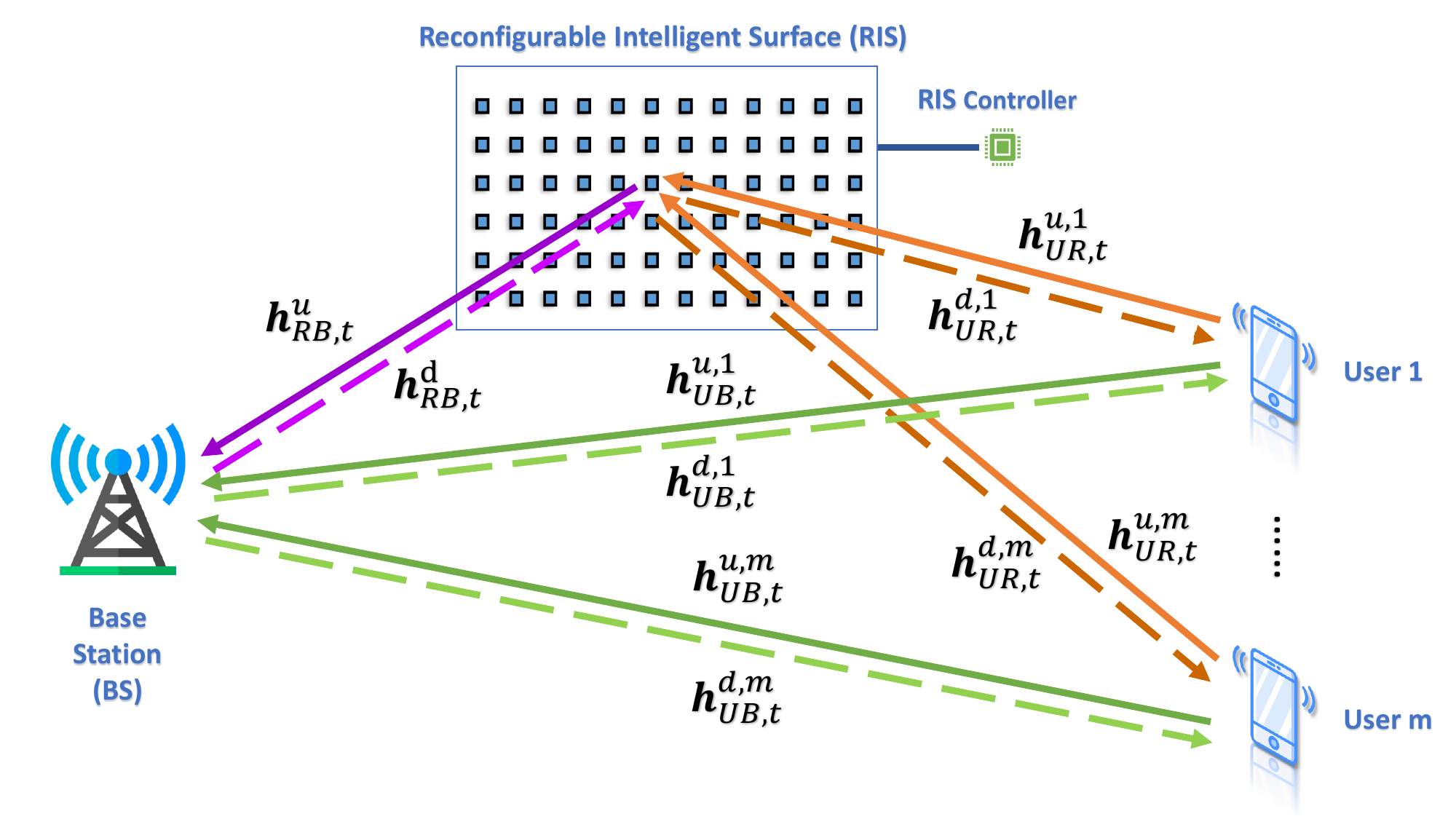}
  \vspace{-15pt}
  \caption{The RIS-assisted communication system.}
  \vspace{-10pt}
  \label{fig:sysmodel}
\end{figure}


As shown in Fig.~\ref{fig:sysmodel} with solid lines, we consider an RIS-assisted uplink communication model\footnote{Without loss of generality, we consider synchronous models.}, where an RIS equipped with $N$ passive elements is embedded between $m$ single-antenna edge devices and a single-antenna PS.
We adopt a block fading channel in which the channel coefficients remain unchanged within one global communication round and vary independently between communication rounds.
In $t$-th round, $h_{UB,t}^{u,i} \in \Cb$ is uplink channel from edge device $i$ to PS, $\h_{RB,t}^{u} \in \Cb^N$ is channel from RIS to PS, and $\h_{UR,t}^{u,i} \in \Cb^N$ is channel from user $i$ to RIS.
We denote RIS $n$-th element as $\theta_{n,t} = e^{j \phi_{n,t}}, n \in [N]$. 
We use a diagonal matrix $\The_t$ to represent the phase shift matrix of the RIS, i.e., $\The_t = diag (\theta_{1,t}, \theta_{2,t},\cdots, \theta_{N,t})$, and update it in the beginning of each global round.
Denote the transmit signal from device $i$ as $\x_t^i$, the PS receives:
\begin{equation}
    \y_t^u = \sum_{i \in [m]} (h_{UB,t}^{u,i} + (\h_{UR,t}^{u,i})^H \The_t \h_{RB,t}^u)\x^i_t + \z_t^u, \label{equ:upreceivesig}
\end{equation}
where $\z_t^u$ is an i.i.d. additive white Gaussian noise (AWGN) vector, whose entries have zero mean and variance $\sigma_{c,u}^2$.
To simplify the notation, we define the equivalent RIS-aided link as $\g_t^{u,i} = ((\h_{UR,t}^{u,i})^H \Hb_{RB,t}^u)^H \in \Cb^N$, where $\Hb_{RB,t}^u = diag(\h_{RB,t}^u)$.
Let $\theb_t=(\theta_{1,t},...,\theta_{N,t})^T$ be the phase vector of RIS.
We further define the effective $i$-th-user-PS channel response in $t$-th round as:
\begin{equation}
    h_t^{u,i}=h_{UB,t}^{u,i} +  (\h_{UR,t}^{u,i})^H \The_t \h_{RB,t}^u.
\end{equation}
Then the corresponding received signal can be written as:
\begin{equation}
\begin{array}{ll}
     \y_t^u & = \sum_{i \in [m]} (h_{UB,t}^{u,i} +  (\g_t^{u,i})^H \theb_t)\x^i_t + \z_t^u   \\
     & = \sum_{i \in [m]} h_t^{u,i} \x^i_t + \z_t^u.
\end{array}
\end{equation}
The transmit power constraint of device $i$ in $t$-th round is:
\begin{equation}
\mathbb{E}[\| \x^i_t \|^2] \leq P_t^i, \forall i \in [m], \forall t, \label{inequ:power}
\end{equation}
where $P_t^i$ is the maximum power budget of device $i$.

\subsubsection{Downlink Communication Model} 
\label{subsubsec: downlink}

As illustrated in Fig.~\ref{fig:sysmodel} with dashed lines, we consider an RIS-assisted downlink communication model.
Similarly, we assume a block fading channel model; the channel gains are invariant within one communication round and change independently in different rounds.
In the $t$-th round, $h_{UB,t}^{d,i} \in \Cb$ is downlink channel from PS to device $i$, $\h_{RB,t}^{d} \in \Cb^N$ is channel from PS to RIS, $\h_{UR,t}^{d,i} \in \Cb^N$ is channel from RIS to device $i$.
Given transmit signal $\x_t^d$, edge device $i$ receives:
\begin{equation}
    \y_t^{d,i} = (h_{UB,t}^{d,i} + (\h_{RB,t}^d)^H \The_t \h_{UR,t}^{d,i})\x_t^d + \z_t^{d,i}, \label{equ:dlrecvsig}
\end{equation}
where $\z_t^{d,i}$ represents an i.i.d. AWGN with zero mean and variance $\sigma_{c,d}^2$.
For ease of notation, we define effective cascaded downlink channel gain of device $i$ as:
\begin{equation}
    h_t^{d,i}=h_{UB,t}^{d,i} + (\h_{RB,t}^d)^H \The_t \h_{UR,t}^{d,i}.
\end{equation}
Therefore, the received signal can be rewritten as: 
\begin{equation}
    \y_t^{d,i} = h_t^{d,i}\x_t^d + \z_t^{d,i}.
\end{equation}
The downlink power constraint is:
\begin{equation}
\mathbb{E}[\| \x^d_t \|^2] \leq P_t^d, \forall t, \label{inequ:dlpower}
\end{equation}
where $P_t^d$ is the maximum power at the PS in round $t$.

In this work, edge devices only have access to imperfect CSI, i.e., users have estimated CSI for both uplink and downlink.
We denote estimated CSI as $\hh_t$ of each path:
\begin{equation}
\hh_t = h_t + \Delta_t, \forall t, \label{equ: estcsi}
\end{equation}
where $\Delta_t$ represents the error of estimation with zero mean and variance $\sigt_h^2$.
Note that all paths have channel estimation error.
We define the overall uplink channel estimate as $\hh_t^{u,i}$, and $\hh_t^{d,i}$ as the overall estimated downlink CSI.

Our single RIS model, though centralized, is capable of addressing personalized federated learning challenges, as will be shown in numerical results.
Additionally, this model can naturally evolve into a distributed personal RISs model, a key extension that we will elaborate in Section~\ref{sec: extension}, showcasing its enhanced applicability in personalized learning scenarios.

\vspace{-5pt}
\section{Joint Communication and Learning Design} 
\label{sec: alg}
In this section, we present a cross-layer approach for joint optimization of communication, computation, and learning resources to enhance system performance. The RIS phase matrix is updated at the start of each global round, followed by dynamic local step design to meet transmit power constraints. 
Meanwhile, we employ dynamic power control schemes for both uplink and downlink communications.
Algorithm~\ref{alg:roar-fed} outlines the case for error-free downlink, while Algorithm~\ref{alg:apaf} addresses the noisy uplink and downlink scenario.

\vspace{-10pt}
\subsection{Power Control for Uplink}
\label{subsec:pcuplink}

\begin{algorithm}[t!] 
    \caption{RIS-assisted Over-the-Air Adaptive Resource Allocation for Federated Learning (\algns)} \label{alg:roar-fed} 
    \begin{algorithmic}[1]
    \STATE 
    \emph{\bf Initialization: global model $\w_0$, $\theb_0$, $\beta_t^i$, $\tau_t^i, i \in [m]$.}
    \FOR{$t=0, \dots, T-1$}
    \STATE {Server first finds the client with the maximum number of local steps in round $t-1$, then applies SCA to compute the RIS phase update:} 
        \FOR{$j=0, \dots, J-1$}
        \STATE {Server updates RIS phase design $\phib_t$ by~\eqref{equ:phase}}.
        \ENDFOR
        \STATE {Server broadcasts the global model $\w_t$.}
        \FOR{each user $i \in [m]$}
        \STATE {Each client finds $\tau_t^i$ to satisfy the power constraint~\eqref{inequ:power} and trains local model by~\eqref{equ:sgd}, then designs $\beta_t^i$ by~\eqref{equ:betaiimp} and transmits $\x_t^i$ by~\eqref{equ:transig}.}   
        \ENDFOR
    \STATE {The server aggregates and updates global model by~\eqref{equ:globalup}.}
    \ENDFOR
    \end{algorithmic}
\end{algorithm}

We employ a dynamic power control (PC) at both PS and edge devices. 
Specifically, in round $t$, device $i$ computes its transmit signal $\x_t^i$ after completing its local training and uploads it to PS.
We introduce adaptive PC factors $\beta_t^i$ and $\beta_t^u$ to adjust the power levels of device $i$ and PS, respectively.
Thereby, transmit signal $\x_t^i$ can be expressed as:
\begin{equation}
    \x_t^i = \beta_t^i (\w^i_{t, \tau_t^i} - \w^i_{t, 0}), \label{equ:transig}
\end{equation}
When the PS receives the aggregated signal~\eqref{equ:upreceivesig}, it scales with PC factor $\beta_t^u$ to update the global model parameter:
\begin{align} \label{equ:globalup}
    \w_{t+1} &= \w_{t} + \frac{1}{\beta_t^u}\sum_{i=1}^{m} h_t^{u,i} \x_t^i  + \tilde{\z}_t^u, 
\end{align}
where $\tilde{\z}_t^u$ is scaled noise whose variance of each entry is $\frac{\sigma_{c,u}^2}{(\beta_t^u)^2}$.

We propose a channel inversion policy combined with dynamic local steps to mitigate channel effects. By adjusting local steps and applying channel inversion to counter fading, we optimize both learning and communication resources simultaneously. The adaptive PC parameter for device $i$ is designed as:
\begin{equation} \label{equ:betaiimp}
    \beta_t^i = \frac{\beta_t^u \alpha_i}{\tau^i_t \hh_t^{u,i}}.
\end{equation}
Note that only estimated CSI $\hh_t^{u,i}$ is available at edge device $i$. Imperfect CSI can cause signal misalignment during aggregation, degrading learning performance. By adjusting local steps based on the current imperfect CSI, we mitigate this issue while maximizing local computation resources.
Additionally, we set another criteria for device $i$:
\begin{equation} \label{inequ:betai}
    3\eta_t^2 \beta_t^i \tau_t^i G^2 \leq P_t^i,
\end{equation}
where $G$ is the upper bound of stochastic gradient, defined in Assumption~\ref{a_bounded} in Sec.~\ref{sec: conv}.
Equation~\eqref{inequ:betai} also aids in designing the RIS phase matrix in round $t$, optimizing the propagation environment of the RIS-assisted system, as discussed in the next subsection. 
Additionally, \eqref{inequ:betai} ensures the convergence of Algorithm~\ref{alg:roar-fed}, as shown in Section~\ref{sec: conv}. 
After phase design, each device selects its local update steps $\tau_t^i$ by meeting the transmit power constraint, as illustrated in Algorithm~\ref{alg:roar-fed}. 

\vspace{-0.1in}
\subsection{Phase Design}
\label{subsec:phase}

As discussed in Section~\ref{subsec:pcuplink}, we design the RIS phase shifts according to~\eqref{inequ:betai}. 
Although a single RIS improves the wireless environment, the constraint~\eqref{inequ:betai} applies to all devices, making it necessary to select a device for whom to update the RIS phase in each round. 
To achieve this, we choose the device with the maximum local steps from the previous iteration and optimize the phase matrix based on its criteria.
By taking this approach, we can alternatively 
ensure reliable and efficient communication for all devices over multiple rounds\footnote{Recall that we have a block fading model.}.
First, we substitute~\eqref{equ:betaiimp} into~(\ref{inequ:betai}):
\begin{equation} \label{inequ:phase}
    (\g_t^i)^H \theb_t \geq \frac{3 \eta_t^2 \beta_t^u \alpha_i G^2}{P_t^i} - \hh_{UB,t}^{u,i}, 
\end{equation}
where $\g_t^i$ and $\theb_t$ are defined in Section~\ref{subsec: comm}.
In order to obtain the desired phase vector $\theb_t$, we formulate the following optimization problem:
\begin{equation} \label{prob:phase}
    \begin{aligned} 
        & \mathop{min}\limits_{\theb_t} \quad \|(\g_t^i)^H \theb_t - \frac{3 \eta_t^2 \beta_t^u \alpha_i G^2}{P_t^i} + \hh_{UB,t}^{u,i} \|_2^2  \\
        &\begin{array}{ll}
        s.t. & |\theta_{t,n}|=1 , \quad n=1,...,N.
        \end{array}
    \end{aligned}
\end{equation}

Note that the constraint of RIS elements renders problem~\eqref{prob:phase} non-convex, complicating the search for an optimal solution.
To address this, we apply a successive convex approximation (SCA) approach to find a stationary point~\cite{scutari2013,mao22papa}.
SCA solves a series of simpler convex problems iteratively by using surrogate objective functions that are differentiable and strongly convex~\cite{scutari2013}, enabling effective approximation of the original non-convex problem.

First, the original objective function is defined:
\begin{equation} \label{equ:phasefunc}
\begin{array}{ll}
     f(\theb_t) & = || s_t^i - (\g_t^i)^H \theb_t||_2^2   \\
     & = (s_t^i)^* s_t^i - 2 Re \{ \theb_t^H \textbf{v}\} + \theb_t^H \textbf{U} \theb_t ,
\end{array}
\end{equation}
where $s_t^i = \frac{3 \eta^2 \beta_t^u \alpha_i G^2}{P_t^i} - \hh_{UB,t}^i$, $\textbf{v} = s_t^i \g_t^i$, $\textbf{U} = \g_t^i (\g_t^i)^H$.
Next, we substitute each phase component with $\theta_{n,t} = e^{j \phi_{n,t}}, \phi_{n,t} \in \Rb$.
Note that, since $s_t^i$ is constant, the problem is equivalent to minimize:
\begin{equation}
    f_1(\phib_t) = (e^{j\phib_t})^H \textbf{U} e^{j\phib_t} - 2 Re\{(e^{j\phib_t})^H \textbf{v}\},
\end{equation}
where $\phib_t = (\phi_{1,t},...,\phi_{N,t})^T$.
Now we adopt the SCA approach. 
In round $j$, we use the second order Taylor expansion to design the surrogate function of $f_1(\phib_t)$ at point $\phib_i^j$:
\begin{equation}
\begin{array}{ll}
      g(\phib_t,\phib^j_t) & = f_1(\phib^j_t) + \nabla f_1(\phib^j_t)^T (\phib_t - \phib^j_t)\\
     &  + \frac{\lambda}{2} ||\phib_t - \phib^j_t||_2^2,
\end{array}
\end{equation}
where $\nabla f_1(\phib^j_t)$ represents the gradient. We choose the parameter $\lambda$ such that the surrogate function satisfies the condition $g(\phib_t,\phib^j_t) \geq f_1(\phib_t) $.
Then we update $\phib_t$ iteratively until convergence by the following rule:
\begin{equation}
    \phib_t^{j+1} = \phib^j_t - \frac{\nabla f_1(\phib^j_t)}{\lambda}. \label{equ:phase}
\end{equation}
After the SCA process is completed, we get the updated RIS phase vector for learning round $t$ as $\theb_t= e^{j\phib_t}$.

We emphasize that dynamic power control enables edge devices to meet communication constraints while exploiting local computation resources. Additionally, we incorporate RIS phase design into each iteration to enhance learning performance under imperfect CSI.

\vspace{-10pt}
\subsection{Power Control for Downlink}
\label{subsec:downpc}
Now, we consider a noisy downlink channel model from Section~\ref{subsubsec: downlink}.
In a noiseless downlink, edge devices would receive the perfect global model during FL broadcasting, leaving only uplink transmission to handle, as shown in Algorithm~\ref{alg:roar-fed}.
However, with a noisy downlink, devices cannot accurately recover the global model.
\begin{algorithm}[t!] 
    \caption{RIS-Assisted Over-the-Air Adaptive Federated Learning with Noisy Downlink} \label{alg:apaf} 
    \begin{algorithmic}[1]
    \STATE 
    \emph{\bf Initialization: $\w_0$, $\theb_0$, $\beta_t^i$, $\tau_t^i, i \in [m]$.}
    \FOR{$t=0, \dots, T-1$}
    \STATE{\bf $\bullet$ Phase update:}
    \STATE {PS selects the user with maximum $\tau_{t-1}^i$ and updates RIS phase shifts by SCA method~\eqref{equ:phase}.} 
    \STATE{\bf $\bullet$ Downlink transmission:}
        \STATE {PS broadcasts global model $\w_t$ by~\eqref{equ:dltransig}, users estimate the starting global model weights by~\eqref{equ:wt0i}.}
    \STATE{\bf $\bullet$ Uplink transmission:}
        \FOR{user $i, i \in [m]$ in parallel}
        \STATE {User $i$ computes $\tau_t^i$ to meet communication constraints and does local training. Then it finds $\beta_t^i$ by~\eqref{equ:betaiimp} and uploads $\x_t^i$.} 
        \ENDFOR
    \STATE {PS updates the global model by~\eqref{equ:globalup}.}
    \ENDFOR
    \end{algorithmic}
\end{algorithm}

In each training round, the PS broadcasts the global model vector to all edge devices via noisy downlink channels.
At the $t$-th iteration, the PS transmits $\x_t^d$ which carries information regarding the global model, with a downlink PC factor $\beta_t^d$.
The broadcasting signal is:
\begin{equation}
    \x_t^d = \beta_t^d \w_{t}, \label{equ:dltransig}
\end{equation}
where $\w_{t}$ is the global model vector.
After receiving ~\eqref{equ:dlrecvsig}, edge device $i$ scales the signal with estimated CSI as following:
\begin{equation} \label{equ:wt0i}
    \hat{\w}_t =  \frac{1}{\beta_t^d \hh_t^{d,i}} \y_t^{d,i} = \frac{h_t^{d,i}}{\hh_t^{d,i}} \w_{t} + \tilde{\z}_t^{d,i},
\end{equation}
where $\tilde{\z}_t^{d,i}$ is the scaled downlink noise.
Then device $i$ does local training.
Note that in this case, each device starts with the global model estimate $\w^i_{t, 0} = \hat{\w}_t$ rather than perfect $\w_t$. 
With imperfect CSI, the estimated signal includes both channel noise and misalignment, potentially degrading learning performance. We will analyze this effect in Section~\ref{sec: conv}.
We summarize the joint learning and communication design for noisy uplink and downlink in Algorithm~\ref{alg:apaf}.

\vspace{-5pt}
\subsection{Computation Complexity}
We begin by analyzing the computational complexity of the SCA algorithm for RIS phase design.
As shown in Algorithm~\ref{alg:roar-fed}, SCA method involves $J$ gradient descent iterations, each updating an $N$-dimensional phase vector, resulting in a complexity of $\mathcal{O}(JN)$.
In the error-free downlink scenario, each global round involves the PS broadcasting the global model, and all edge devices receive perfect parameters, resulting in a complexity of $\mathcal{O}(1)$.
Each device performs local training via SGD and computes the transmit signal in parallel, leading to an uplink transmission complexity of $\mathcal{O}(\max(\tau_t^i)q + q)$, where $q$ is the model parameter dimension.
Model aggregation has a complexity of $\mathcal{O}(1)$ due to OTA-FL's simultaneous communication and computation. Therefore, the overall complexity of Algorithm~\ref{alg:roar-fed} is $\mathcal{O}(T(JN + 1 + \max(\tau_t^i)q + q))$, where $T$ is the total number of iterations.
In the noisy downlink scenario, additional complexity arises from the downlink power control as each user estimates the global model vector. Since all users perform~\eqref{equ:wt0i} independently,  the complexity remains $\mathcal{O}(q)$.
Combining this with the phase design and uplink transmission, the total complexity of Algorithm~\ref{alg:apaf} is $\mathcal{O}(T(JN + q + \max(\tau_t^i)q + q))$.

\vspace{-0.1in}
\section{Convergence Analysis}
\label{sec: conv}
In this section, we formally analyze the learning performance of our cross-layer design.
We start by introducing standard assumptions of stochastic optimization. 
Using these, we intuitively explain how wireless network properties affect learning performance by analyzing convergence upper bound components. 
We illustrate the convergence of Algorithm~\ref{alg:apaf}, noting that Algorithm~\ref{alg:roar-fed} can be regarded as a special case.

We consider a series of general non-convex local loss function with three standard assumptions:
\begin{assum}($L$-Smoothness) \label{a_smooth}
	 For device $i \in [m]$, the gradient of its loss function meets the requirement $ \| \nabla F_i(\w_1) - \nabla F_i(\w_2) \| \leq L \| \w_1 - \w_2 \|$, $\forall \w_1, \w_2 \in \mathbb{R}^q$.
\end{assum}
\begin{assum}(Unbiased Local Stochastic Gradients with Bounded Variance) \label{a_unbias}
	The stochastic local gradient is both unbiased and has a bounded variance, i.e.,
	$\mathbb{E} [\nabla F_i(\w, \xi_i)] = \nabla F_i(\w)$, and $\mathbb{E} [\| \nabla F_i(\w, \xi_i) -  \nabla F_i(\w) \|^2] \leq \sigma^2$, $\forall i \in [m]$, where $\xi_i$ is a random sample in $D_i$.
\end{assum}
\begin{assum}(Bounded Stochastic Gradient) \label{a_bounded}
	The local stochastic gradients have bounded norms, i.e., $\mathbb{E} [\| \nabla F_i(\w, \xi_i) \|^2] \leq G^2$, $\forall i \in [m]$.
\end{assum}
We provide the convergence analysis in Theorem~\ref{thm:convergence}:
\begin{restatable}[Convergence Analysis] {theorem} {convergence} \label{thm:convergence}
    Let the learning rate be a constant, i.e., $\eta_t = \eta  \leq \frac{1}{L}$. Set $P_t^i=P_i, \forall t \in [T]$. Under Assumptions~\ref{a_smooth}-~\ref{a_bounded}, we have:
    \begin{multline}
        \min_{t \in [T]} \mb{E} \| \nabla F(\w_t) \|^2 \leq \underbrace{\frac{2 \left(F(\w_0) - F(\w_{*}) \right)}{T \eta}}_{\mathrm{optimization \, error}} + \underbrace{\frac{L \sigma_{c,u}^2}{ \eta \beta^2}}_{\substack{\mathrm{channel \,
        noise} \\ \mathrm{error}}} \nonumber \\
        + \underbrace{ \frac{2 m L^2}{9 \eta G^2}  \sum_{i=1}^m  \frac{(\alpha_i)^2 P_i^2}{(\beta_i^2)} }_{\mathrm{local \, update \, error}} + \underbrace{  L \eta \sigma^2 \frac{1}{T} \sum_{t=0}^{T-1} \sum_{i=1}^{m} \alpha_i^2 \mb{E}_t  \bigg\| \frac{h_t^{u,i}}{\hh_t^{u,i}} \bigg\|^2}_{\mathrm{statistical \, error}}  \nonumber \\
          + \underbrace{2 m G^2  \frac{1}{T} \sum_{t=0}^{T-1} \sum_{i=1}^m (\alpha_i)^2 \mb{E}_t  \bigg\| 1 - \frac{h_t^{u,i}}{\hh_t^{u,i}} \bigg\|^2}_{\mathrm{uplink \, channel \, estimation \, error}} \nonumber \\
        + \underbrace{2 m L^2 \frac{1}{T} \sum_{t=0}^{T-1} \sum_{i=1}^m (\alpha_i)^2 \left( \mb{E}_t \bigg\|1 - \frac{h_{t}^{d,i}}{\hh_{t}^{d,i}} \bigg\|^2 V(t) \right) }_{\mathrm{downlink \, channel \, estimation \, error}} \nonumber \\
        + \underbrace{2 m L^2 \frac{1}{T} \sum_{t=0}^{T-1} \sum_{i=1}^m(\alpha_i)^2\frac{\sigma_{c,d}^2}{\| \hh_{t}^{d,i}\|^2 (\beta_t^d)^2}}_{\mathrm{downlink \, channel \, noise \, error}} \nonumber \\
    \end{multline}
    \vspace{-5pt}
    where $\frac{1}{\beta_i^2} = \frac{1}{T} \sum_{t=0}^{T-1} \frac{1}{(\beta_t^i)^2}$, $\frac{1}{\bar{\beta}^2} = \frac{1}{T} \sum_{t=0}^{T-1} \frac{1}{(\beta_t^u)^2}$, and
    \begin{equation*}
        V(t) \triangleq 2 \|\w_0\|^2 + \sum_{l=0}^{t-1} \frac{\sigma_{c,u}^2}{(\beta_l^u)^2} + 2 t m G^2 \eta^2\sum_{l=0}^{t-1}  \sum_{i=1}^m \alpha_i^2 \left( \frac{h_l^{u,i}}{\hh_l^{u,i}}\right)^2.
    \end{equation*}
\end{restatable}

\begin{proof}[Proof] 
See Appendix A.
\end{proof}
\vspace{-0.1in}
The convergence upper bound is affected by seven resources of error, as shown in Theorem~\ref{thm:convergence}: 1) FL optimization error; 2) uplink noise; 3) local update error; 4) statistical error; 5) uplink channel estimation error; 6) downlink channel estimation error; 7) downlink noise.
In the error-free downlink scenario, the last two terms disappear. 
These errors are interconnected due to the cross-layer design, and with only imperfect CSI at edge devices, both uplink and downlink estimation errors impact convergence.
With perfect CSI, two estimation errors will vanish, consistent with the literature that fading channels can be fully mitigated by power control.
Furthermore, with the assistance of RIS, our algorithms can achieve excellent learning performance even when the direct links are weak.

We analyze the term $\frac{h_t^i}{\hat{h_t^i}}$ to further bound the statistical and channel estimation errors in Theorem~\ref{thm:convergence}.
Using the Taylor expansion from~\cite{zhu2020one}, we approximate it: $\frac{h_t^i}{\hh_t^i} = \frac{1}{1 + \frac{\Delta_t^i}{h_t^i}}  = 1 - \frac{\Delta_t^i}{h_t^i} + \mc{O}( (\frac{\Delta_t^i}{h_t^i})^2)$. 
Neglecting higher-order terms and assuming i.i.d. CSI estimation errors and unit norm RIS phase elements, we obtain the following result:
\begin{restatable}{corollary} {convergence_rate} \label{cor:convergence}
Let $|\Delta_t| \ll |h_t|, \forall t\in[T]$, $h_{UB,m} = \mathop{min}\limits_{t \in [T], i \in[m]}\{|h_{UB,t}^{i}|\}$, $h_{UR,a} = \mathop{max}\limits_{t \in [T], i \in[m], j \in [N]}\{|h_{UR,t,j}^{i}|\}$, $h_{RB,a} = \mathop{max}\limits_{t \in [T], j \in [N]}\{|h_{RB,t,j}|\}$, $g_m=\mathop{min}\limits_{t \in [T], i \in[m], j \in [N]}\{|g_{t,j}^i|\}$ for both downlink and uplink, the convergence rate of Algorithm~\ref{alg:apaf} is bounded. The statistical error and channel estimation errors are bounded by:
\begin{align*}
        & L \eta \sigma^2 \frac{1}{T} \sum_{t=0}^{T-1} \sum_{i=1}^{m} \alpha_i^2 \mb{E}_t  \bigg\| \frac{h_t^{u,i}}{\hh_t^{u,i}} \bigg\|^2  \leq 
        L \eta \sigma^2 \sum_{i=1}^{m} \alpha_i^2 \left( 1 + C\right), \nonumber \\
        & 2 m G^2  \frac{1}{T} \sum_{t=0}^{T-1} \sum_{i=1}^m (\alpha_i)^2 \mb{E}_t  \bigg\| 1 - \frac{h_t^{u,i}}{\hh_t^{u,i}} \bigg\|^2  \leq 2 m G^2 \sum_{i=1}^m (\alpha_i)^2 C, \nonumber \\
        & 2 m L^2 \frac{1}{T} \sum_{t=0}^{T-1} \sum_{i=1}^m (\alpha_i)^2 \left( \mb{E}_t \bigg\|1 - \frac{h_{t}^{d,i}}{\hh_{t}^{d,i}} \bigg\|^2 V(t) \right) \nonumber \\
        & \quad \quad \quad \quad \quad \quad \quad \quad \leq 2 m L^2 \sum_{i=1}^m(\alpha_i)^2CV(T), 
\end{align*}
where $C=\frac{\sigt_h^2 (1+N^2(h_{UR,a}^2+h_{RB,a}^2+\sigt_h^2))}{(h_{UB,m})^2}$. 
\end{restatable}
\begin{proof}[Proof] 
See Appendix B.
\end{proof}

The convergence bound is impacted by the number of RIS elements $N$, as shown in Corollary~\ref{cor:convergence}. 
While increasing $N$ raises the channel estimation error due to more links, this error is not dominant.
More RIS elements improve channel conditions, reducing local update and downlink noise errors, and ultimately enhancing the overall learning performance.
Specifically, as demonstrated in~\eqref{inequ:risimp} within Appendix~\ref{sec:profcoro}, an increase in $N$ increases the channel coefficient, leading to a better convergence bound, as demonstrated in Section~\ref{sec: exp}.

\vspace{-0.4in}
\section{Extension to Personalization}
\label{sec: extension}
\vspace{-0.2in}
In this section, we extend our framework to personalized federated learning with personal RIS. 
Specifically, we leverage the potential of personal RIS in enhancing OTA-FL within 6G programmable wireless environments.
We introduce our cross-layer approach to such distributed personal RIS model to assist PFL for individual clients, the proposed framework is summarized in Algorithm~\ref{alg:PROR-fed}.

We consider a personal RIS-assisted communication system, as shown in Fig.~\ref{fig:sysmodel_ext}. 
Now, each client is associated with a dedicated RIS, deployed nearby and exclusively controlled by the client itself.
We investigate a bi-level optimization for PFL:
\begin{equation}
    \begin{aligned} 
        & \min_{\mathbf{v}^i \in \mathbb{R}^q}R_i(\mathbf{v}^i;\w^*) \triangleq F_i(\mathbf{v}^i, D_i) + \frac{\lambda}{2} \|\mathbf{v}^i - \w^*\|^2  \\
        &\begin{array}{ll}
        s.t. & \w^* = \underset{\w}{\mathrm{arg\,min}} F(\w),
        \end{array}
    \end{aligned}
    \label{equ: localobjective}
\end{equation}
where $\mathbf{v}^i$ is personalized local model parameter, and $\lambda$ is the hyperparameter for regularization~\cite{li2021ditto}. 
We assume an error-free downlink in this extension, noting that the noisy downlink can be handled by the same power control as in Sec.~\ref{subsec:downpc}. 
Notably, each user updates their own RIS, eliminating the need for user selection as in single RIS setups.

\begin{figure}[t] 
    \centering
    \includegraphics[scale=0.25]{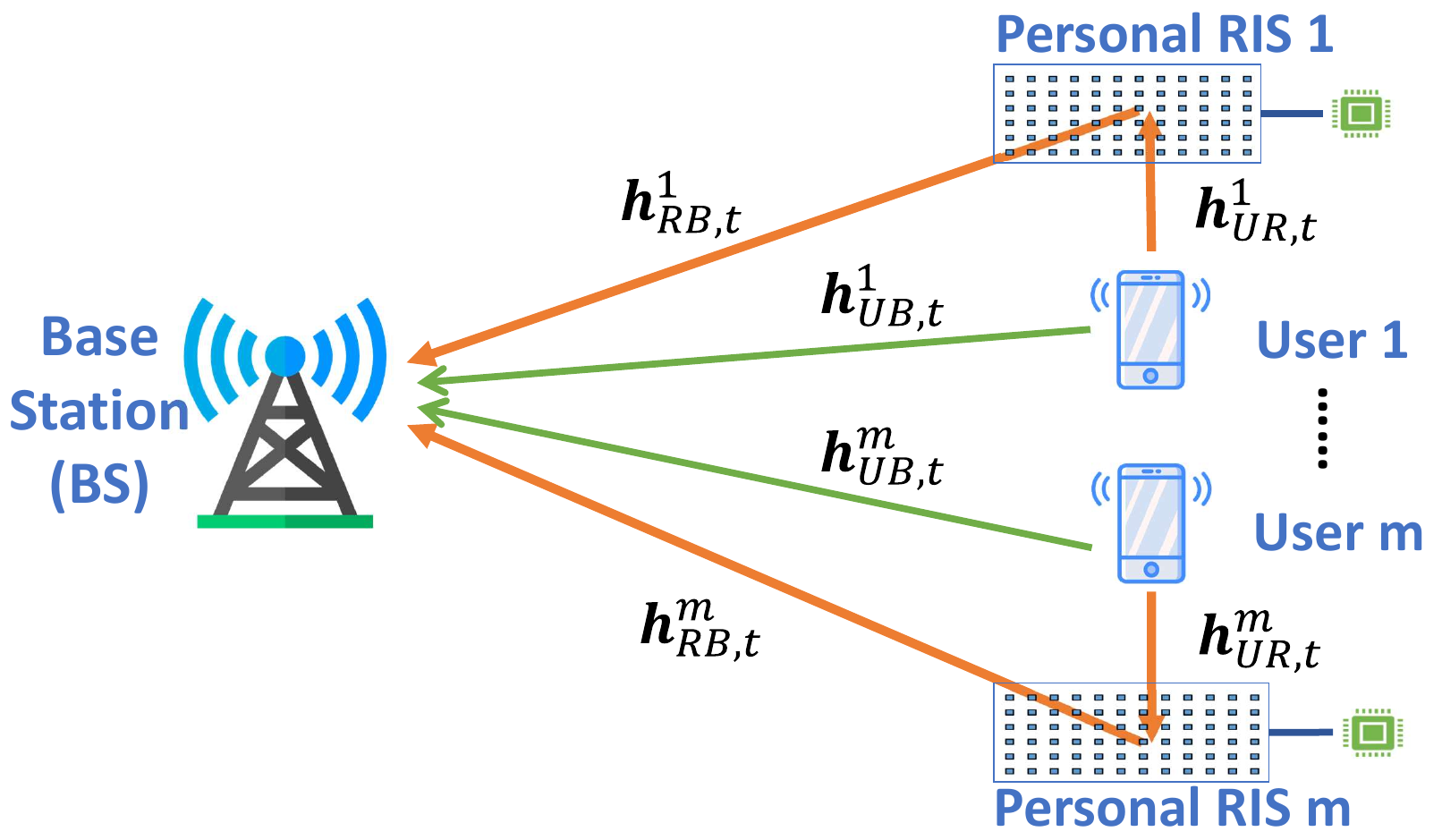}
    \vspace{-2pt}
    \caption{The personal RIS-assisted communication system.}
    \label{fig:sysmodel_ext}
    \vspace{-0.05in}
\end{figure}

\begin{algorithm}[t!] 
    \caption{Personal RIS-assisted Over-the-Air Resource Allocation for Personalized Federated Learning~(\algp)} \label{alg:PROR-fed} 
    \begin{algorithmic}[1]
    \STATE 
    \emph{\bf Initialization: $\theb_0^i$, $\w_0$, $\mathbf{v}_0^i$, $\beta_t^i$, $\tau_v^i$, $\tau_t^i,i \in [m]$.}
    \FOR{$t=0, \dots, T-1$}
        \FOR{each user $i \in [m]$ in parallel}    \STATE {User $i$ updates its own RIS iteratively by~\eqref{equ:phase}}.
        \ENDFOR
        \STATE {PS broadcasts $\w_t$ to all users.}
        \FOR{each user $i \in [m]$ in parallel}
        \STATE {Each user obtains $\tau_t^i$ for power constraint, conducts local training, calculates $\beta_t^i$ by~\eqref{equ:betaiimp} and transmits $\x_t^i$.}
        \STATE {Each user performs $\tau_v^i$ local steps to update $\mathbf{v}_t^i$.}    
        \ENDFOR
    \STATE {The PS aggregates and updates global model by~\eqref{equ:globalup}.}
    \ENDFOR
    \STATE {\textbf{return} $\{\mathbf{v}^i\}_{i \in [m]} (personalized), \w_T (global)$}
    \end{algorithmic}
\end{algorithm}
\textbf{Convergence of Algorithm~\ref{alg:PROR-fed}.} The global model $\w$ is independent of the personalized models $\{\mathbf{v}^i\}_{i \in [m]}$. Therefore, the convergence rate of the global optimization is the same as \alg.
\vspace{-5pt}
\begin{assum}\label{a_unbias_local}
    (Unbiased Local Stochastic Gradients)
    $\mathbb{E} [g_i(\mathbf{v}^i;\w)] = \nabla R_i(\mathbf{v}^i;\w)$, where $g_i(\mathbf{v}^i;\w)$ is stochastic gradient of $R_i(\mathbf{v}^i;\w)$.
\end{assum}
\vspace{-5pt}
Given these observations, we are now able to provide an analysis of the convergence for personalized local tasks.
\vspace{-5pt}
\begin{restatable} {theorem} {convergence} \label{thm:conv_per}
    With Assumptions~\ref{a_smooth}-~\ref{a_unbias_local}, a constant global learning rate $\eta_t = \eta  \leq \frac{1}{L}$, a constant local learning rate $\eta_v \leq \frac{1}{\sqrt{2L^2 + 2 \lambda^2}}$ and $T\geq 4$, for each device $i\in [m]$, we have:
    \vspace{-10pt}
    \begin{multline}
        \min_{t \in [T]} \mb{E} \| \nabla R_i(\mathbf{v}_t^i;\w^{*}) \|^2 \leq 
        \underbrace{\sqrt{2L^2 + 2 \lambda^2} \eta_v^2 \sigma^2}_{\mathrm{statistical \, error}}  \nonumber \\
        +\underbrace{ \frac{1}{T} \sum_{s=0}^{T-1}\frac{2 \left(R_i(\mathbf{v}_s^i;\w_s) - R_i(\mathbf{v}_{s+1}^{i};\w_s) \right)}{\tau_v(\eta_v - \frac{\sqrt{2L^2 + 2 \lambda^2}}{2}\eta_v^2)} }_{\mathrm{optimization \, error}}   \underbrace{+ \lambda^2 T^2 \frac{\sigma_{c,u}^2}{\beta^2}}_{\mathrm{channel \, noise \, error}} \nonumber \\
        + \underbrace{2 \lambda^2 m G^2 \eta^2\frac{1}{T} \sum_{s=0}^{T-1}  (T-s) \sum_{l=s}^T \sum_{i=1}^m \alpha_i^2 \mb{E}_t \bigg\| \frac{h_l^{u,i}}{\hh_l^{u,i}} \bigg\|^2}_{\mathrm{global \, model \, update \, error}} \nonumber
    \end{multline}
    where $\frac{1}{\bar{\beta}^2} = \frac{1}{T} \sum_{t=0}^{T-1} \frac{1}{(\beta_t^u)^2}$, $\tau_v = \tau_v^i$.
\end{restatable}
\begin{proof}[Proof] 
See Appendix C.
\end{proof}
\vspace{-10pt}
Using the same analytical method of $\frac{h_l^{u,i}}{\hh_l^{u,i}}$, we derive Corollary~\ref{cor:convergence}, showing that the convergence bound is finite.
\begin{restatable}{corollary} {convergence_rate} \label{cor:conv_per}
Let $|\Delta_t| \ll |h_t|, \forall t\in[T]$, $h_{UB,m} = \\
\mathop{min}\limits_{t \in [T], i \in[m]}\{|h_{UB,t}^i|\}$, $h_{UR,a} = \mathop{max}\limits_{t \in [T], i \in[m], j \in [N]}\{|h_{UB,t,j}^i|\}$, $h_{RB,a} = \mathop{max}\limits_{t \in [T], j \in [N]}\{|h_{RB,t,j}|\}$, $\eta = \frac{1}{T}$, $\beta = T$, $\tau_v^i = \tau_v$, $\alpha_i=\frac{1}{m}$, $\exists \lambda < \epsilon, \epsilon>0,$ the convergence rate is bounded:
\begin{align}
        &\min_{t \in [T]} \mb{E} \| \nabla R_i(\mathbf{v}_t^i;\w^{*}) \|^2 \leq \sqrt{2L^2 + 2\lambda^2} \eta_v^2 \sigma^2 + 2 \lambda^2 G^2(1+C)
         \nonumber \\
        &+ \frac{1}{T} \sum_{s=0}^{T-1}\frac{2 \left(R_i(\mathbf{v}_s^i;\w_s) - R_i(\mathbf{v}_{s+1}^{i};\w_s) \right)}{\tau_v(\eta_v - \frac{\sqrt{2L^2 + 2 \lambda^2}}{2}\eta_v^2)}  + \lambda^2 \sigma_{c,u}^2, \nonumber
\end{align}
where $C=\frac{\sigt_h^2 (1+N^2(h_{UR,a}^2+h_{RB,a}^2+\sigt_h^2))}{(h_{UB,m})^2}$. 
\end{restatable}
\begin{proof}[Proof] 

We use the result of $\mb{E}_t \bigg\| \frac{h_l^{u,i}}{\hh_l^{u,i}} \bigg\|^2$ in Appendix C. 
Given sufficiently small  $\lambda$, the regularization term's impact on convergence and boundedness is minimal, ensuring the optimization error remains bounded.
\end{proof}

\vspace{-0.2in}
\section{Numerical Results} 
\label{sec: exp}
\subsection{Set Up}
We conduct numerical experiments to evaluate the effectiveness of our cross-layer algorithms and compare them with the state-of-the-art.
We set up an RIS-assisted multiuser communication system.
Specifically, the system consists of $m=10$ users and a single RIS.
Following~\cite{liu2021risfl}, we consider a 3D coordinate system, where the location of BS is (-50,0,10) meters and RIS is deployed at (0,0,10) meters.
We distribute the users uniformly in the x-y plane within the range of [-20,0] meters in the x dimension and [-30,30] meters in the y dimension.
The channel model follows i.i.d. Gaussian distribution with scaling of the square root of path loss.
We consider the path loss model in~\cite{tang2020ris}.
The user-PS direct link model is  $G_{PS}G_{U}\left(\frac{3*10^8 m/s}{4 \pi f_c d_{UP}} \right)^{PL}$, where we set path loss exponent $PL=4$ to simulate weak direct link, $G_{PS}=5$dBi, $G_U=0$dBi are antenna gains, $d_{UP}$ is the user-PS distance, and $f_c=915$MHz is the carrier frequency.
The RIS assisted link model is $G_{PS}G_{U} G_{RIS} \frac{N^2 d_x d_y ((3*10^8 m/s)/f_c)^2 }{64 \pi^3 d_{RP}^2 d_{UR}^2}$, where $d_x=d_y=(3*10^7 m/s)/f_c$ are the dimensions of an RIS element, $G_{RIS}=5$dBi is the RIS antenna gain, and $d_{RP}, d_{UR}$ are RIS-PS distance, user-RIS distance, respectively.
To simulate an imperfect CSI scenario, we define the channel estimation error of each path as i.i.d. Gaussian random variable with variance $\sigt_h^2= 0.1 \sigma_{c,u}^2$.
The transmit uplink SNR is set to $20$dB, while the transmit downlink SNR is set to $30$dB since the server is usually power-sufficient.

\vspace{-0.1in}
\subsection{Results}
We consider a standard image classification task on the MNIST~\cite{lecun1998gradient} and Fashion-MNIST~\cite{xiao2017fashion} datasets, both widely used to evaluate FL algorithms.
We begin with MNIST, using RIS with $N=45$ elements for classification. A logistic regression model, consisting of 784 input neurons and 10 output neurons, is trained in an extreme non-i.i.d. setting where each local dataset contains data from only one class. The training parameters are: local epochs $E=1$, batch size $BS=30$, and learning rate $\eta=0.1$.

We next use Fashion-MNIST, using $N=100$ RIS elements to aid with this more complex dataset. A convolutional neural network (CNN) model is trained with two $5 \times 5$ convolution layers (10 and 20 channels), followed by $2 \times 2$ max pooling, a batch normalization layer, a fully connected layer with 50 units and ReLU activation, and a softmax output layer. The training parameters are $E=1$, $BS=64$, $\eta=0.05$, decaying every 50 rounds at 0.5. We use SGD with momentum 0.9 and weight decay 0.001. We assume a label-imbalanced non-i.i.d. setting, where data is distributed across edge devices using a Dirichlet distribution with concentration parameter $\gamma$, as a common choice to simulate real-world data distribution~\cite{li2022federated}. Initially, $\gamma=0.5$, then adjusted to 1 to reduce the non-i.i.d. level. Hyperparameters are fine-tuned with $E=1$, $BS=64$, and $\eta={0.1, 0.05}$.

We compare our proposed algorithms with several baselines:
\begin{list}{\labelitemi}{\leftmargin=1em \itemindent=-0.5em \itemsep=.2em}
\item Baseline 1: Algorithm from~\cite{li2022one}. A RIS is deployed to aid one-bit communication for FL, i.e., one-bit gradient quantization. The sub-band assignment, power allocation and RIS configuration are jointly optimized.
\item Baseline 2: Algorithm from~\cite{liu2021risfl}. A RIS is positioned between the server and the users to assist OTA-FL, where a unified design is proposed that includes active devices, RIS phase matrix and receiver beamforming. 
\end{list}

\begin{figure}[t] 
    \centering
    \includegraphics[width=0.85\linewidth]{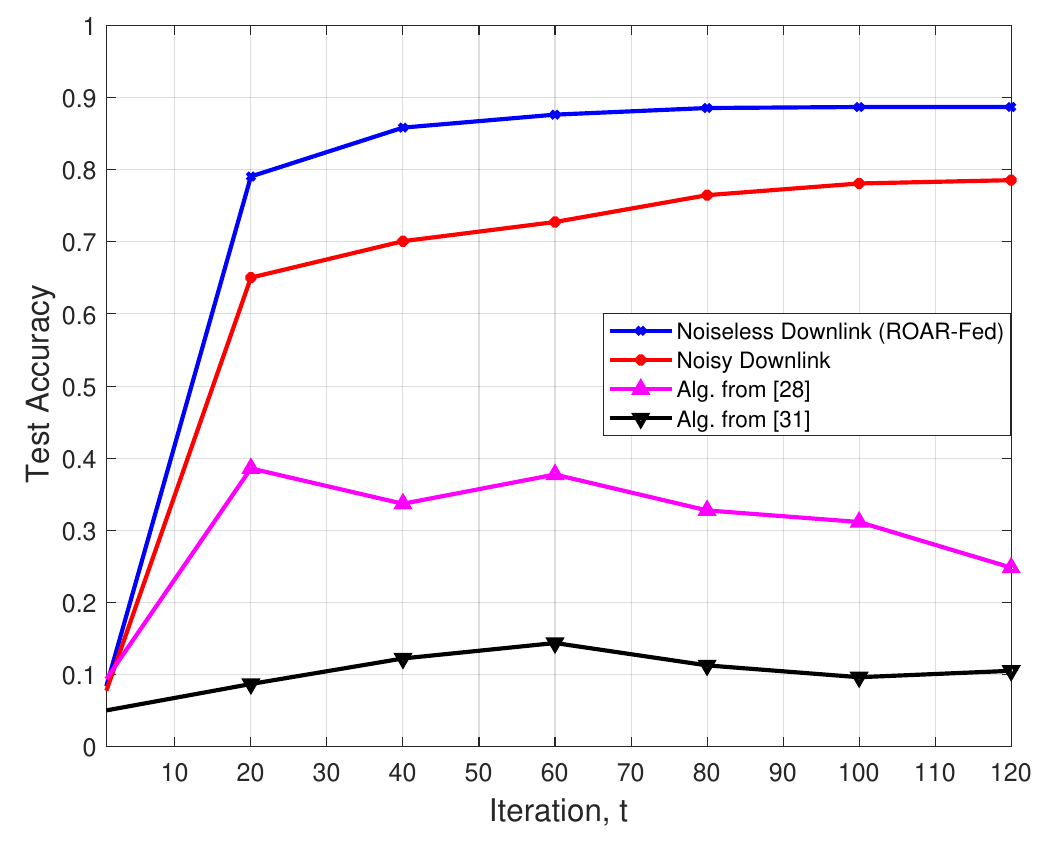}
    \vspace{-5pt}
    \caption{Test accuracy on the MNIST dataset.}
    \vspace{-10pt}
    \label{fig:result1}
\end{figure}

\begin{figure}[t] 
    \centering
    \includegraphics[width=0.85\linewidth]{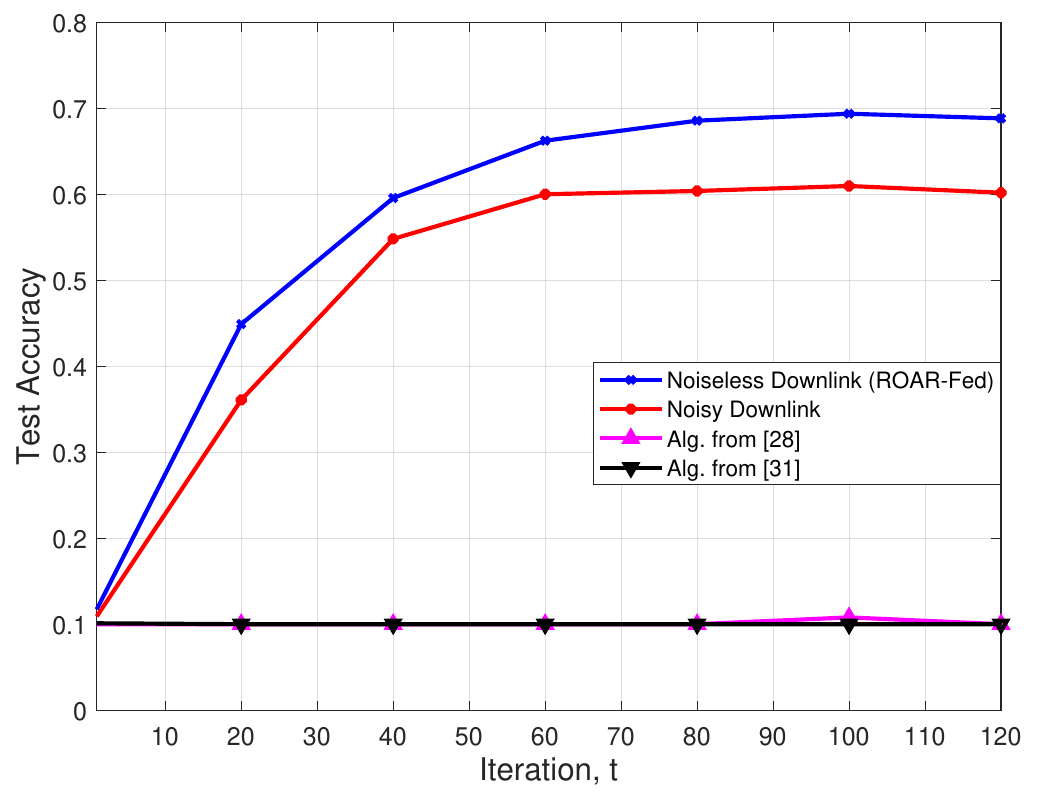}
    \vspace{-5pt}
    \caption{Test accuracy on Fashion-MNIST with $\gamma = 0.5$.}
    \vspace{-5pt}
    \label{fig:result2}
\end{figure}

\begin{figure}[t] 
    \centering
    \includegraphics[width=0.85\linewidth]{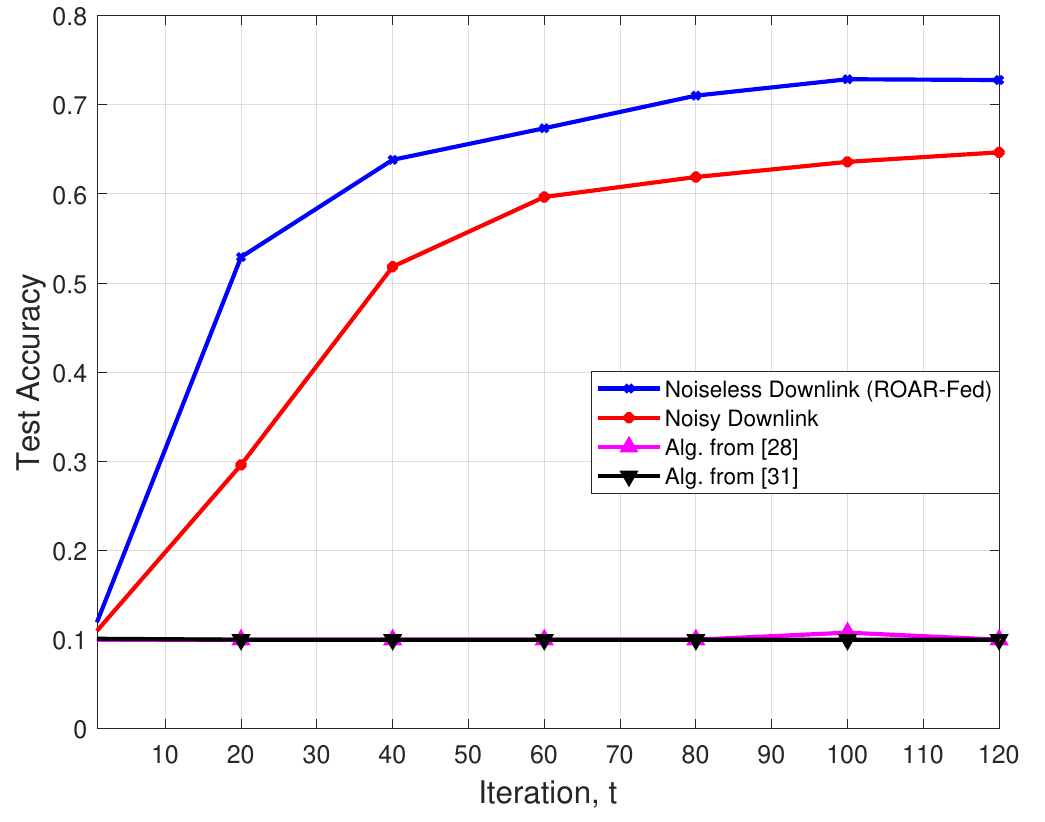}
    \vspace{-5pt}
    \caption{Test accuracy on Fashion-MNIST with $\gamma = 1$.}
    \vspace{-5pt}
    \label{fig:result3}
\end{figure}

Fig.~\ref{fig:result1} presents the test accuracy results on the MNIST dataset of the proposed Algorithm~\ref{alg:roar-fed}, Algorithm~\ref{alg:apaf} as well as the benchmark algorithms versus the number of communication rounds.
Similarly, Fig~\ref{fig:result2}. shows the test accuracy on the Fashion-MNIST dataset.
Notably, Algorithm~\ref{alg:apaf} attains outstanding learning performance even with both noisy uplink and downlink with imperfect CSI on both datasets.
The learning performance of Algorithm~\ref{alg:apaf} experiences some degradation, compared to the noiseless downlink case (\algns), as anticipated.
This highlights that noisy downlink has a nontrivial negative influence on learning performance, which matches our theoretical analysis. 
It is worth mentioning that both baseline algorithms~\cite{li2022one,liu2021risfl} fail to converge under such realistic conditions, especially in noisy time-varying channels with imperfect communication and non-i.i.d. data distribution.
The one-bit approach~\cite{li2022one} performs slightly better than the algorithm from~\cite{liu2021risfl} on the MNIST dataset, but it still fails to achieve good learning accuracy as shown in Fig.~\ref{fig:result1}.
In addition, none of them can get effective training on the Fashion-MNIST dataset in such a general setting.
By contrast, our algorithms exhibit both convergence and superior learning performance, thus verifying the effectiveness and robustness of our adaptive joint learning and communication design.
One can also infer that our approach achieves a degree of fairness, as it ensures an enhanced propagation environment and training performance for all edge devices in an equitable, iterative and distributed manner by the coupled design of local computation and RIS configuration.
Note that a larger number of local steps leads to a global model that is more biased towards that specific device and our choice of selection coupled with i.i.d. block fading model ensures fairness.
Finally, a comparison between Fig.\ref{fig:result2} and Fig.\ref{fig:result3} provides insights into the impact of data heterogeneity levels. As $\gamma$ increases, corresponding to a reduction in the non-i.i.d. level, an enhancement in the learning performance is observed for both noisy and error-free downlink scenarios, which aligns with intuition.

\begin{figure}[t] 
    \centering
    \includegraphics[width=0.85\linewidth]{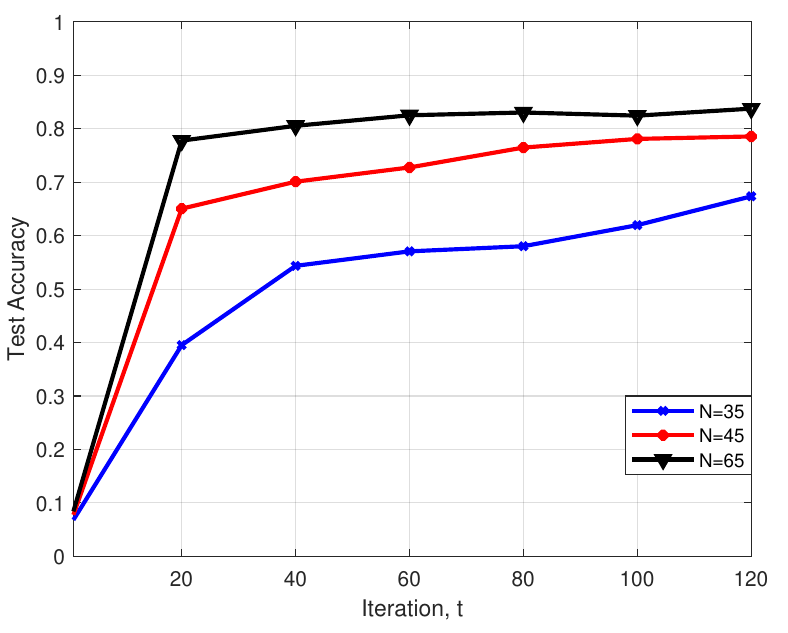}
    \vspace{-5pt}
    \caption{Test accuracy versus the number of RIS elements.}
    \vspace{-5pt}
    \label{fig:result4}
\end{figure}

Fig.~\ref{fig:result4} demonstrates the influence of the number of RIS reflecting elements $N$ on the MNIST dataset.
We notice that as $N$ increases, both the test accuracy and the rate of convergence improve, which aligns with our theoretical findings presented in Section~\ref{sec: conv}.
As $N$ grows, the channel estimation error increases, as illustrated in Corollary 1. 
However, the convergence upper bound is not dominated by this error term.
Instead, the local update error plays a more important role, which reduces with an increased $N$.
This is because a larger $N$ enables a better radio environment, which results in fewer local steps and a more substantial power control factor overall.

\vspace{-0.2in}
\subsection{Extension to Personalization}
We next consider a system with personal RISs, still having $m=10$ users, and each RIS has $N=10$ elements. Each RIS is deployed 2 meters above its owner. Other system settings remain the same. We set additional parameters as $\lambda=0.1$ and local steps $\tau_v^i=3$. 
We conduct image classification tasks using the Fashion-MNIST dataset, characterized by a non-i.i.d. data distribution, following a Dirichlet distribution with parameters $\gamma=0.1$ and $0.5$. For this task, we employ the same CNN model as used in our main results. We consider two benchmarks: 1) \algns, a central RIS at $(0, 0, 10)$ meters with $m \times N = 100$~\cite{ni2021fed} elements; 2) \alghota~\cite{mortaheb2022personalized}, a personalized OTA-FL method in a hierarchical structure without RIS. The system contains PS, intermediate servers, and clients grouped into clusters, with error-free PS-server links but wireless fading between servers and clients. For a fair comparison, we modify this by equating cluster and client counts to $m$, leading to one-to-one intermediate server-client communication. This aligns each client with a wireless fading channel, comparable to our system's configuration.
This approach also differs in its learning structure: each client's neural network comprises several shared global layers topped with a personalized layer. Its evaluation focuses solely on personalization accuracy, not on global task performance.

\begin{figure}[t] 
    \centering
    \includegraphics[width=0.8\linewidth]{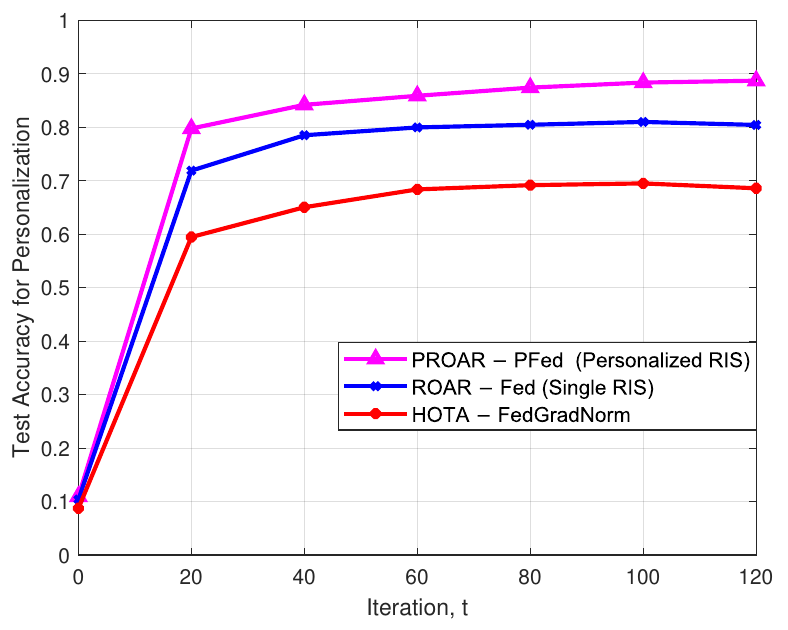}
    \vspace{-5pt}
    \caption{Average test accuracy for personalization with $\gamma=0.5$.}
    \vspace{-5pt}
    \label{fig:result_ext}
\end{figure}

Fig.~\ref{fig:result_ext} illustrates the average test accuracy across 10 users for personalized tasks when $\gamma=0.5$.
\algp~outperforms benchmarks, highlighting the advantage of personal RISs. This confirms the effectiveness of our joint communication and learning design using bi-level optimization, and suggests that smaller personalized RISs are preferable to a single large centralized RIS, especially in handling device heterogeneity or multi-task FL.

\begin{table}[t]
\caption{Average (5 trails) of the CNN test accuracy (\%) comparison on Fashion-MNIST with various non-i.i.d. level $\gamma$. We report \textbf{average} across users for personalized models. “N/A” means not applicable.}
\label{table:test_accuracy}
\centering
{\scriptsize
\begin{tabular}{|c|c|cc|}
\hline
\multicolumn{1}{|c|}{\multirow{2}{*}{\textbf{Non-IID}}} & \multicolumn{1}{c|}{\multirow{2}{*}{\textbf{Algorithm}}} & \multicolumn{2}{c|}{\textbf{FL Model}}                                        \\ \cline{3-4} 
\multicolumn{1}{|c|}{}                               & \multicolumn{1}{c|}{}                           & \multicolumn{1}{c|}{Global}  & \multicolumn{1}{c|}{Personalized}    \\ \hline
\multirow{3}{*}{$\gamma = 0.1$}                                 & {\cellcolor[gray]{.9}{\textbf{\algp} }}                                     & \multicolumn{1}{c|}{\cellcolor[gray]{.9} {\bf 63.08}} & \multicolumn{1}{c|}{\cellcolor[gray]{.9} {\bf 96.04 $\pm$ 2.69}}  \\ \cline{2-4} 
                                                     & \alg                                           & \multicolumn{1}{c|}{57.68} & \multicolumn{1}{c|}{88.97 $\pm$ 4.74}  \\ \cline{2-4} 
                                                     & \alghota                                       & \multicolumn{1}{c|}{N/A} & \multicolumn{1}{c|}{89.37 $\pm$ 6.97} \\ \hline
\multirow{3}{*}{$\gamma = 0.5$}                                 & {\cellcolor[gray]{.9}{\bf \algp}}                                     & \multicolumn{1}{c|}{\cellcolor[gray]{.9} {\bf 71.89}} & \multicolumn{1}{c|}{\cellcolor[gray]{.9} {\bf 88.89 $\pm$ 3.33}} \\ \cline{2-4} 
                                                     & \alg                                           & \multicolumn{1}{c|}{68.77} & \multicolumn{1}{c|}{80.84 $\pm$ 5.56} \\ \cline{2-4} 
                                                     & \alghota                                       & \multicolumn{1}{c|}{N/A} & \multicolumn{1}{c|}{69.77 $\pm$ 3.89} \\ \hline
\end{tabular}
}
\end{table}

Table~\ref{table:test_accuracy} shows the results of different non-i.i.d. cases.
Note that smaller $\gamma$ indicates a more unbalanced data partition.
Consequently, the performance of the global model for the $\gamma=0.1$ case is worse compared to the $\gamma=0.5$ scenario.
Notably, \alghota~has the lowest accuracy when $\gamma=0.5$, yet it surpasses \algns~when $\gamma=0.1$. 
This indicates that the design for personalized FL is more effective in scenarios with greater data heterogeneity, even in the absence of RIS assistance.
Compared with \algns, \algp~improves
global model learning, and personalized tasks outperform the global model significantly.
Furthermore, the performance difference between the personalized and global models grows as the data sets are more non-i.i.d., demonstrating that \algp~is better suited for handling highly heterogeneous situations. 

\vspace{-10pt}
\section{Conclusion} 
\label{sec: conclusion}

In this paper, we have proposed a cross-layer approach that jointly optimizes computation and communication resources to enhance learning performance in RIS-assisted OTA-FL systems.
Unlike existing joint designs, we have considered the practical time-varying physical layer, i.e., uplink and downlink noise, and imperfect CSI at edge devices.
The developed method adaptively adjusts RIS phase shifts, local update steps, and transmission power during each global iteration, mitigating the impacts of time-varying estimated CSI and non-i.i.d local data distribution.
We have provided a comprehensive convergence proof of the proposed algorithms, discussed sources of error for convergence upper bounds, and analyzed the impact of the number of RIS reflecting elements.
We have extended our framework to the personal RIS model and illustrated its capability to personalized federated learning.
Finally, we have demonstrated the effectiveness and robustness of our proposed approach through numerical results, which outperforms the state-of-the-art unified communication and learning frameworks in realistic scenarios with data heterogeneity and imperfect CSI.

\appendices
\input{proof_submit.tex}
\vspace{-10pt}
\bibliographystyle{IEEEtran}{}
\bibliography{BIB/Optimization,BIB/Relay, BIB/RIS, BIB/ImperfectCSI, BIB/Yener, BIB/RISFL, BIB/Experiments, BIB/Introduction}

\end{document}

%% file: symbols_commands.tex
\newcommand{\Cb}{\mathbb{C}}

\newcommand{\Eb}{\mathbb{E}}

\newcommand{\g}{\mathbf{g}}
\newcommand{\Hb}{\mathbf{H}}
\newcommand{\h}{\mathbf{h}}
\newcommand{\hh}{\widehat{h}}

\newcommand{\Rb}{\mathbb{R}}

\newcommand{\sigt}{\widetilde{\sigma}}

\renewcommand{\v}{\mathbf{v}}

\newcommand{\w}{\mathbf{w}}

\newcommand{\x}{\mathbf{x}}

\newcommand{\y}{\mathbf{y}}
\newcommand{\z}{\mathbf{z}}

\newcommand{\mc}[1]{\mathcal{#1}}
\newcommand{\mb}[1]{\mathbb{#1}}

\newtheorem{assum}{Assumption}

\newcommand{\The}{\mathbf{\Theta}}
\newcommand{\theb}{\boldsymbol{\theta}}
\newcommand{\phib}{\boldsymbol{\phi}}



%% file: proof_submit.tex
\allowdisplaybreaks

\vspace{-8pt}
\section{Proof of Theorem 1}
\label{sec:proof}
The global model update in $t$-th round is:
\begin{align}
    &\w_{t+1} - \w_{t} = \sum_{i=1}^{m} \frac{\beta_t^i}{\beta_t^u} h_t^{u,i} \left(\w^i_{t, \tau^i_t} - \w^i_{t, 0}\right) + \tilde{\z}_t \\
    &= \sum_{i=1}^{m} \frac{\alpha_i}{\tau^i_t} \frac{h_t^{u,i}}{\hh_t^{u,i}} \left(\w^i_{t, \tau^i_t} - \w^i_{t, 0}\right) + \tilde{\z}_t \\
    &= -\sum_{i=1}^{m} \frac{\alpha_i}{\tau^i_t}\frac{h_t^{u,i}}{\hh_t^{u,i}} \eta_t \sum_{k=0}^{\tau^i_t-1} \left(\nabla F_i(\w^i_{t, k}, \xi^i_{t, k})\right) + \tilde{\z}_t.
\end{align}
We first take conditional expectation with respect to $\w_t$.
Assuming that stochastic gradient noise, channel noise and channel estimation error are independent, by Assumption 1:
\begin{align}
    &\mb{E}_t [F(\w_{t+1})] - F(\w_t) \leq \nonumber \\
    & \left< \nabla F(\w_t), \mb{E}_t \left[\w_{t+1} - \w_t \right] \right>  + \frac{L}{2} \mb{E}_t \left[\| \w_{t+1} - \w_t \|^2 \right] \\
    &= \frac{L}{2} \mb{E}_t \bigg\| - \sum_{i=1}^{m} \eta_t\frac{\alpha_i }{\tau^i_t}\frac{h_t^{u,i}}{\hh_t^{u,i}} \sum_{k=0}^{\tau^i_t-1} \left(\nabla F_i(\w^i_{t, k}, \xi^i_{t, k})\right)  + \tilde{\z}_t \bigg\|^2 \nonumber\\
    &- \left< \nabla F(\w_t), \mb{E}_t \left[\sum_{i=1}^{m} \frac{\alpha_i}{\tau^i_t} \frac{h_t^{u,i}}{\hh_t^{u,i}} \eta_t \sum_{k=0}^{\tau^i_t-1} \left(\nabla F_i(\w^i_{t, k})\right)\right] \right> \\
    & = \frac{L}{2} \mb{E}_t \bigg\| - \sum_{i=1}^{m} \frac{\alpha_i \eta_t}{\tau^i_t} \frac{h_t^{u,i}}{\hh_t^{u,i}} \sum_{k=0}^{\tau^i_t-1} \left(\nabla F_i(\w^i_{t, k}, \xi^i_{t, k})\right) + \tilde{\z}_t \bigg\|^2 + \nonumber \\
    & \eta_t \left< \nabla F(\w_t), \nabla F(\w_t) - \mb{E}_t \left[\sum_{i=1}^{m} \frac{\alpha_i}{\tau^i_t}\frac{h_t^{u,i}}{\hh_t^{u,i}} \sum_{k=0}^{\tau^i_t-1} \left(\nabla F_i(\w^i_{t, k})\right) \right] \right> \nonumber\\
    & - \eta_t \| \nabla F(\w_t) \|^2 \\
    &= - \frac{1}{2} \eta_t \mb{E}_t \bigg\| \sum_{i=1}^{m} \frac{\alpha_i}{\tau^i_t} \frac{h_t^{u,i}}{\hh_t^{u,i}} \sum_{k=0}^{\tau^i_t-1} \left(\nabla F_i(\w^i_{t, k})\right) \bigg\|^2  \nonumber \\
    & + \frac{1}{2} \eta_t \mb{E}_t \bigg\| \nabla F(\w_t) - \sum_{i=1}^{m} \frac{\alpha_i}{\tau^i_t} \frac{h_t^{u,i}}{\hh_t^{u,i}}\sum_{k=0}^{\tau^i_t-1} \left(\nabla F_i(\w^i_{t, k})\right) \bigg\|^2 \nonumber\\
    &+ \frac{L}{2} \mb{E}_t \bigg\| - \sum_{i=1}^{m} \frac{\alpha_i \eta_t}{\tau^i_t}\frac{h_t^{u,i}}{\hh_t^{u,i}} \sum_{k=0}^{\tau^i_t-1} \left(\nabla F_i(\w^i_{t, k}, \xi^i_{t, k})\right) + \tilde{\z}_t \bigg\|^2 \nonumber \\
    &- \frac{1}{2} \eta_t \| \nabla F(\w_t) \|^2\\
    &=- \frac{1}{2} \eta_t \| \nabla F(\w_t) \|^2  + \frac{L \sigma_c^2}{2\beta_t^2} \nonumber \\
    & - \frac{1}{2} \eta_t \mb{E}_t \bigg\| \sum_{i=1}^{m} \frac{\alpha_i}{\tau^i_t} \frac{h_t^{u,i}}{\hh_t^{u,i}} \sum_{k=0}^{\tau^i_t-1} \left(\nabla F_i(\w^i_{t, k})\right) \bigg\|^2 \nonumber \\
    &+ \frac{1}{2} \eta_t \mb{E}_t \bigg\| \sum_{i=1}^{m} \frac{\alpha_i}{\tau^i_t} \sum_{k=0}^{\tau^i_t-1} \left(\nabla F_i(\w_t) - \frac{h_t^{u,i}}{\hh_t^{u,i}} \nabla F_i(\w^i_{t, k})\right) \bigg\|^2  \nonumber \\
    &+ \frac{L \eta_t^2}{2} \mb{E}_t \bigg\| \sum_{i=1}^{m} \frac{\alpha_i}{\tau^i_t} \frac{h_t^{u,i}}{\hh_t^{u,i}}\sum_{k=0}^{\tau^i_t-1} \left(\nabla F_i(\w^i_{t, k}, \xi^i_{t, k})\right) \bigg\|^2 \\
    &\leq \frac{1}{2} \eta_t \mb{E}_t \bigg\| \sum_{i=1}^{m} \frac{\alpha_i}{\tau^i_t} \sum_{k=0}^{\tau^i_t-1} \left(\nabla F_i(\w_t) - \frac{h_t^{u,i}}{\hh_t^{u,i}} \nabla F_i(\w^i_{t, k})\right) \bigg\|^2 + \nonumber \\ 
    &\frac{L \eta_t^2}{2} \sum_{i=1}^{m} \mb{E}_t \bigg\| \frac{\alpha_i}{\tau^i_t}\frac{h_t^{u,i}}{\hh_t^{u,i}}  \sum_{k=0}^{\tau^i_t-1} \left(\nabla F_i(\w^i_{t, k}, \xi^i_{t, k}) - \nabla F_i(\w^i_{t, k})\right) \bigg\|^2 \nonumber \\ & - \frac{1}{2} \eta_t \| \nabla F(\w_t) \|^2 + \frac{L \sigma_c^2}{2\beta_t^2} \label{ineq: Lsmooth} 
\end{align}
The first inequality holds if $\eta_t \leq \frac{1}{L}$. Since learning process and channel estimation are independent,
\begin{align}
    &\frac{ L \eta_t^2}{2} \sum_{i=1}^{m} \mb{E}_t \bigg\| \frac{\alpha_i}{\tau^i_t}\frac{h_t^{u,i}}{\hh_t^{u,i}}  \sum_{k=0}^{\tau^i_t-1} \left(\nabla F_i(\w^i_{t, k}, \xi^i_{t, k}) - \nabla F_i(\w^i_{t, k})\right) \bigg\|^2  \leq \nonumber\\
    &
    \sum_{i=1}^m \frac{ L \eta_t^2(\alpha_i)^2}{2(\tau^i_t)^2} \mb{E}_t  \bigg\| \frac{h_t^{u,i}}{\hh_t^{u,i}} \bigg\|^2 \times \nonumber \\
    &\mb{E}_t \bigg\|  \sum_{k=0}^{\tau^i_t-1} \left(\nabla F_i(\w^i_{t, k}, \xi^i_{t, k}) - \nabla F_i(\w^i_{t, k})\right) \bigg\|^2 \\
    &\leq \sum_{i=1}^m \frac{L \eta_t^2 (\alpha_i)^2}{2 \tau^i_t} \mb{E}_t  \bigg\| \frac{h_t^{u,i}}{\hh_t^{u,i}} \bigg\|^2 \times \nonumber \\
    & \sum_{k=0}^{\tau^i_t-1} \mb{E}_t \bigg\| \nabla F_i(\w^i_{t, k}, \xi^i_{t, k}) - \nabla F_i(\w^i_{t, k}) \bigg\|^2 \\
    &\leq \frac{ L \eta_t^2}{2} \sum_{i=1}^m (\alpha_i)^2 \mb{E}_t  \bigg\| \frac{h_t^{u,i}}{\hh_t^{u,i}} \bigg\|^2 \sigma^2
     \label{ineq: boundedvar}
\end{align}
Using Jensen's Inequality to Assumption 2 and 3, we get $\bigg\| \nabla F_i(\w^i_{t, k}) \bigg\|^2 \leq G^2$. 
\begin{align}
    &\frac{1}{2} \eta_t \mb{E}_t \bigg\| \sum_{i=1}^{m} \frac{\alpha_i}{\tau^i_t} \sum_{k=0}^{\tau^i_t-1} \left(\nabla F_i(\w_t) - \frac{h_t^{u,i}}{\hh_t^{u,i}}\nabla F_i(\w^i_{t, k})\right) \bigg\|^2 \nonumber \\
    &\leq \frac{1}{2} \eta_t m \sum_{i=1}^m \frac{\alpha_i^2}{(\tau^i_t)^2} \mb{E}_t \bigg\|  \sum_{k=0}^{\tau^i_t-1} \nabla F_i(\w_t) - \frac{h_t^{u,i}}{\hh_t^{u,i}} \nabla F_i(\w^i_{t, k}) \bigg\|^2 \\
    &\leq \frac{1}{2} \eta_t m \sum_{i=1}^m \frac{(\alpha_i)^2}{\tau^i_t} \sum_{k=0}^{\tau^i_t-1} \mb{E}_t \bigg\| \nabla F_i(\w_t) - \frac{h_t^{u,i}}{\hh_t^{u,i}}\nabla F_i(\w^i_{t, k}) \bigg\|^2 \\
    &= \frac{1}{2} \eta_t m \sum_{i=1}^m \frac{(\alpha_i)^2}{\tau^i_t} \sum_{k=0}^{\tau^i_t-1} \mb{E}_t \bigg\| \nabla F_i(\w_t) -\nabla F_i(\w^i_{t, k}) \nonumber \\
    &  + \nabla F_i(\w^i_{t, k}) - \frac{h_t^{u,i}}{\hh_t^{u,i}}\nabla F_i(\w^i_{t, k}) \bigg\|^2 \\
    &\leq \frac{1}{2} \eta_t m \sum_{i=1}^m \frac{(\alpha_i)^2}{\tau^i_t} \sum_{k=0}^{\tau^i_t-1} \mb{E}_t  \bigg( 2\bigg\| \nabla F_i(\w_t) -\nabla F_i(\w^i_{t, k}) \bigg\|^2 \nonumber \\
    &+ 2 \bigg\| \nabla F_i(\w^i_{t, k}) - \frac{h_t^{u,i}}{\hh_t^{u,i}}\nabla F_i(\w^i_{t, k}) \bigg\|^2 \bigg) \\
    &\leq \eta_t m \sum_{i=1}^m \frac{(\alpha_i)^2}{\tau^i_t} \sum_{k=0}^{\tau^i_t-1} \bigg( L^2 \mb{E}_t \bigg\|\w_t - \w^i_{t, k} \bigg\|^2 \nonumber \\
    &+ \mb{E}_t \bigg\|1 - \frac{h_t^{u,i}}{\hh_t^{u,i}}\bigg\|^2 \mb{E}_t \bigg\| \nabla F_i(\w^i_{t, k}) \bigg\|^2 \bigg)
    \label{ineq: variance}
\end{align}
Note that noisy downlink results in a noisy starting point for each user in each training round. 
\begin{align}
    & \mb{E}_t \bigg\|\w_t - \w^i_{t, k} \bigg\|^2 = \mb{E}_t \bigg\|\w_t - \w_{t,0}^i + \w_{t,0}^i - \w^i_{t, k} \bigg\|^2  \\
    & = \mb{E}_t \bigg\|\w_t - \w_{t,0}^i \bigg\|^2 + \mb{E}_t \bigg\|\w_{t,0}^i - \w^i_{t, k} \bigg\|^2  \\
    & = \mb{E}_t \bigg\| \left( 1 - \frac{h_t^{d,i}}{\hh_t^{d,i}}\right)\w_t + \tilde{\z}_t^{d,i} \bigg\|^2 + \mb{E}_t \bigg\|\w_{t,0}^i - \w^i_{t, k} \bigg\|^2 \\
    & = \mb{E}_t \bigg\| \left( 1 - \frac{h_t^{d,i}}{\hh_t^{d,i}}\right)\w_t \bigg\|^2 + \mb{E}_t \bigg\| \tilde{\z}_t^{d,i} \bigg\|^2 \nonumber \\
    & + \eta_t^2 \mb{E}_t \bigg\|\sum_{j=0}^{k} \nabla F_i(\w^i_{t, j}, \xi^i_{t, j})  \bigg\|^2 \\
    & \leq \mb{E}_t \bigg\|1 - \frac{h_t^{d,i}}{\hh_t^{d,i}} \bigg\|^2 \mb{E}_t \|\w_t\|^2 + \frac{\sigma_{c,d}^2}{(\beta_t^d)^2\| \hh_t^{d,i}\|^2} + \eta_t^2 k^2 G^2
    \label{equ:noisydl}
\end{align}
Next, we compute the bound of $\mb{E}_t \|\w_t\|^2$.
\begin{align}
    \w_t = \w_0 + \sum_{l=0}^{t-1}\sum_{i=1}^m \frac{\beta_l^i}{\beta_l^u}h_l^{u,i}  \sum_{k=0}^{\tau^i_l-1} \left(-\eta_l\nabla F_i(\w^i_{l, k}, \xi^i_{l, k})\right) + \sum_{l=0}^{t-1} \tilde{\z}_l^u
\end{align}
By Cauchy–Schwarz inequality, we can get
\begin{align}
    & \mb{E}_t \|\w_t\|^2 \leq 2 \|\w_0\|^2 + \sum_{l=0}^{t-1} \frac{\sigma_{c,u}^2}{(\beta_l^u)^2} \nonumber \\
    &  2\mb{E} \|\sum_{l=0}^{t-1}\sum_{i=1}^m \frac{\beta_l^i}{\beta_l^u}h_l^{u,i}  \sum_{k=0}^{\tau^i_l-1} \left(-\eta_l\nabla F_i(\w^i_{l, k}, \xi^i_{l, k})\right)\|^2  \\
    & \leq  2 t \sum_{l=0}^{t-1} \eta_l^2  \sum_{i=1}^m (\frac{\beta_l^i}{\beta_l^u}h_l^{u,i})^2 \sum_{i=1}^m \mb{E} \|\sum_{k=0}^{\tau^i_l-1} \left(\nabla F_i(\w^i_{l, k}, \xi^i_{l, k})\right)\|^2 \nonumber \\
    & + 2 \|\w_0\|^2  + \sum_{l=0}^{t-1} \frac{\sigma_{c,u}^2}{(\beta_l^u)^2} \\
    & \leq 2 t \sum_{l=0}^{t-1} \eta_l^2  \sum_{i=1}^m (\frac{\beta_l^i}{\beta_l^u}h_l^{u,i})^2 \sum_{i=1}^m \tau^i_l \sum_{k=0}^{\tau^i_l-1} \mb{E} \|\left(\nabla F_i(\w^i_{l, k}, \xi^i_{l, k})\right)\|^2  \nonumber \\
    & + 2 \|\w_0\|^2 + \sum_{l=0}^{t-1} \frac{\sigma_{c,u}^2}{(\beta_l^u)^2} \\
    & \leq 2 \|\w_0\|^2 + 2 t \sum_{l=0}^{t-1} \eta_l^2  \sum_{i=1}^m (\frac{\beta_l^i}{\beta_l^u}h_l^{u,i})^2 m (\tau_l^i)^2 G^2 + \sum_{l=0}^{t-1} \frac{\sigma_{c,u}^2}{(\beta_l^u)^2} \nonumber \\
    & = 2 \|\w_0\|^2 + \sum_{l=0}^{t-1} \frac{\sigma_{c,u}^2}{(\beta_l^u)^2} + 2 t m G^2 \sum_{l=0}^{t-1} \eta_l^2 \sum_{i=1}^m \alpha_i^2 \left( \frac{h_t^{u,i}}{\hh_t^{u,i}}\right)^2 \nonumber
\end{align}
Define
\begin{align}
    V(t) = 2 \|\w_0\|^2 + \sum_{l=0}^{t-1} \frac{\sigma_{c,u}^2}{\beta_l^2} + 2 t m G^2 \sum_{l=0}^{t-1} \eta_l^2 \sum_{i=1}^m  \alpha_i^2 \left( \frac{h_t^{u,i}}{\hh_t^{u,i}}\right)^2
    \label{equ:noisydlvt}
\end{align}
Then we have $\mb{E}_t \|\w_t \|^2 \leq V(t)$.
Plug~\eqref{equ:noisydlvt} to~\eqref{equ:noisydl},
\begin{align}
    \mb{E}_t \bigg\|\w_t - \w^i_{t, k} \bigg\|^2\leq \mb{E}_t \bigg\|1 - \frac{h_t^{d,i}}{\hh_t^{d,i}} \bigg\|^2 V(t) + \frac{\sigma_{c,d}^2}{(\beta_t^d)^2 \| \hh_t^{d,i}\|^2} + \eta_t^2 k^2 G^2 \label{inequ:noisydlw}
\end{align}
Now, apply constraint~\eqref{inequ:betai} to~\eqref{ineq: variance}, we have
\begin{align}
    &\frac{1}{2} \eta_t \mb{E}_t \bigg\| \sum_{i=1}^{m} \frac{\alpha_i}{\tau^i_t} \sum_{k=0}^{\tau^i_t-1} \left(\nabla F_i(\w_t) - \frac{h_t^{u,i}}{\hh_t^{u,i}}\nabla F_i(\w^i_{t, k})\right) \bigg\|^2 \nonumber\\
    & \leq \eta_t m \sum_{i=1}^m \frac{(\alpha_i)^2}{\tau^i_t} \sum_{k=0}^{\tau^i_t-1} \bigg( L^2 
     \mb{E}_t \bigg\|1 - \frac{h_t^{d,i}}{\hh_t^{d,i}} \bigg\|^2 V(t)  \nonumber \\
     & +\frac{L^2 \sigma_{c,d}^2}{(\beta_t^d)^2 \| \hh_t^{d,i}\|^2} + L^2\eta_t^2 k^2 G^2 + \mb{E}_t \bigg\|1- \frac{h_t^{u,i}}{\hh_t^{u,i}}\bigg\|^2 G^2 \bigg)   \\
    & \leq \eta_t m \sum_{i=1}^m (\alpha_i)^2 L^2 \left( \mb{E}_t \bigg\|1 - \frac{h_t^{d,i}}{\hh_t^{d,i}} \bigg\|^2 V(t) + \frac{\sigma_{c,d}^2}{(\beta_t^d)^2 \| \hh_t^{d,i}\|^2} \right) \nonumber \\
    & + \eta_t^3 m L^2 \sum_{i=1}^m (\alpha_i)^2 (\tau^i_t)^2 G^2 + \eta_t m \sum_{i=1}^m (\alpha_i)^2 \mb{E}_t \bigg\|1- \frac{h_t^{u,i}}{\hh_t^{u,i}}\bigg\|^2 G^2 \nonumber\\
    & \leq \eta_t m \sum_{i=1}^m (\alpha_i)^2 L^2 \left( \mb{E}_t \bigg\|1 - \frac{h_t^{d,i}}{\hh_t^{d,i}} \bigg\|^2 V(t) + \frac{\sigma_{c,d}^2}{(\beta_t^d)^2 \| \hh_t^{d,i}\|^2} \right) \label{inequ:risimp} \\ 
    & + \frac{m L^2}{9 \eta_t G^2} \sum_{i=1}^m \left( \frac{\alpha_i P_t^i}{\beta_t^i}\right)^2 + \eta_t m \sum_{i=1}^m (\alpha_i)^2 \mb{E}_t \bigg\|1- \frac{h_t^{u,i}}{\hh_t^{u,i}}\bigg\|^2 G^2
    \label{ineq:ROAR-fed-noisydl}
\end{align}
Then plug inequality~\eqref{ineq: boundedvar} and~\eqref{ineq:ROAR-fed-noisydl} into \eqref{ineq: Lsmooth}, we have 
\begin{align}
    &\mb{E}_t [F(\w_{t+1})] - F(\w_t) \leq \left< \nabla F(\w_t), \mb{E}_t \left[\w_{t+1} - \w_t \right] \right> \nonumber \\ 
    & + \frac{L}{2} \mb{E}_t \left[\| \w_{t+1} - \w_t \|^2 \right] \\
    &\leq - \frac{1}{2} \eta_t \| \nabla F(\w_t) \|^2 + \frac{m L^2}{9 \eta_t G^2} \sum_{i=1}^m \left( \frac{\alpha_i P_t^i}{\beta_t^i}\right)^2   + \frac{L \sigma_{c,u}^2}{2(\beta_t^u)^2}\nonumber \\
    & + \eta_t m \sum_{i=1}^m (\alpha_i)^2 \bigg( \mb{E}_t \bigg\|1- \frac{h_t^{u,i}}{\hh_t^{u,i}}\bigg\|^2 G^2 + L^2 \mb{E}_t \bigg\|1 - \frac{h_t^{d,i}}{\hh_t^{d,i}} \bigg\|^2 V(t)   \nonumber \\
    &+ \frac{L^2\sigma_{c,d}^2}{(\beta_t^d)^2 \|\hh_t^{d,i}\|^2} \bigg) + \frac{L \eta_t^2}{2} \sum_{i=1}^m (\alpha_i)^2 \mb{E}_t  \bigg\| \frac{h_t^{u,i}}{\hh_t^{u,i}} \bigg\|^2 \sigma^2
\end{align}
Let $\eta_t = \eta$ be constant learning rate, $P_t^i = P_i$, $\frac{1}{\beta_i^2} = \frac{1}{T} \sum_{t=0}^{T-1} \frac{1}{(\beta_t^i)^2}$, after rearranging and telescoping, we obtain:
\begin{align}
    & \frac{1}{T} \sum_{t=0}^{T-1} \mb{E} \| \nabla F(\w_t) \|^2 \leq \frac{2 \left(F(\w_0) - F(\w_T) \right)}{T \eta} \nonumber \\
    & + \frac{2 m L^2}{9 \eta^2 G^2}  \sum_{i=1}^m  \frac{(\alpha_i)^2 P_i^2}{(\beta_i^2)} + \frac{L \sigma_{c,u}^2}{\eta \beta^2} \nonumber \\
    & + 2 m G^2  \frac{1}{T} \sum_{t=0}^{T-1} \sum_{i=1}^m (\alpha_i)^2 \mb{E}_t  \bigg\| 1 - \frac{h_t^{u,i}}{\hh_t^{u,i}} \bigg\|^2 \nonumber \\
    & +  2 m L^2 \frac{1}{T} \sum_{t=0}^{T-1} \sum_{i=1}^m (\alpha_i)^2 \mb{E}_t \bigg\|1 - \frac{h_t^{d,i}}{\hh_t^{d,i}} \bigg\|^2 V(t) \nonumber \\
    & + 2 m L^2 \frac{1}{T} \sum_{t=0}^{T-1} \sum_{i=1}^m (\alpha_i)^2 \frac{\sigma_{c,d}^2}{(\beta_t^d)^2 \| \hh_t^{d,i}\|^2} \nonumber \\
    & + L \eta \frac{1}{T} \sum_{t=0}^{T-1} \sum_{i=1}^m (\alpha_i)^2 \mb{E}_t  \bigg\| \frac{h_t^{u,i}}{\hh_t^{u,i}} \bigg\|^2 \sigma^2,
\end{align}
where $\frac{1}{\bar{\beta}^2} = \frac{1}{T} \sum_{t=0}^{T-1} \frac{1}{(\beta_t^u)^2}$. \\

\vspace{-0.3in}
\section{Proof of Corollary 1}
\label{sec:profcoro}
\vspace{-0.2in}
\begin{align}
    h_t^{u,i} = h_{UB,t}^i+(\h_{UR,t}^i)^H \The_t \h_{RB,t}
\end{align}
\begin{align}
    &\hh_t^{u,i}= h_{t}^{u,i} + \Delta_{UB,t} + (\h_{UR,t}^i)^H \The_t \boldsymbol{\Delta}_{RB,t}  \nonumber \\
    &+ (\boldsymbol{\Delta}_{UR,t}^i)^H \The_t \h_{RB,t} + (\boldsymbol{\Delta}_{UR,t}^i)^H \The_t \boldsymbol{\Delta}_{RB,t}
\end{align}
\begin{align}
    \frac{h_t^{u,i}}{\hh_t^{u,i}} & = \frac{1}{1 + \frac{\hh_t^{u,i} - h_t^{u,i}}{h_t^{u,i}}} 
\end{align}
\begin{equation}
    \text{Define }\Delta_t^i = \hh_t^{u,i} - h_t^{u,i}, A = \frac{\Delta_t^i}{ h_t^{u,i}}
\end{equation}
Use Taylor expansion, $\frac{h_t^{u,i}}{\hh_t^{u,i}} = \frac{1}{1 + \frac{\Delta_t^i}{h_t^{u,i}}} = 1 - \frac{\Delta_t^i}{h_t^{u,i}} + \mathcal{O}( (\frac{\Delta_t^i}{h_t^{u,i}})^2)$, and we can ignore the higher order $\mathcal{O}( (\frac{\Delta_t^i}{h_t^{u,i}})^2)$. Then,

\begin{equation}
    \bigg\| \frac{h_t^{u,i}}{\hh_t^{u,i}} \bigg\|^2 = \bigg\|1 - \frac{\Delta_t^i}{h_t^{u,i}} \bigg\|^2 = 1-2A + \|A\|^2
\end{equation}
Note that $\Eb[A] = 0$ because each path's estimation is i.i.d. with zero mean. Next, we focus on $\Eb[\|A\|^2]$.
\begin{align}
    \|A\|^2 & = \bigg \| \frac{\Delta_t^i}{ h_{UB,t}^i + (h_{UR,t}^i)^H diag (\h_{RB,t}) \theb_t}\bigg \|^2 \\
    & =   \frac{\|[\h_{UR,t}^i diag(\boldsymbol{\Delta}_{RB,t})+(\boldsymbol{\Delta}_{UR,t}^i)^H diag(\h_{RB,t})}{ \| h_{UB,t}^i + (\g_t^i)^H \theb_t \|^2} \nonumber \\
    &\frac{+ \Delta_{UB,t} +(\boldsymbol{\Delta}_{UR,t}^i)^H diag(\boldsymbol{\Delta}_{RB,t})]\theb_t\|^2}{ \| h_{UB,t}^i + (\g_t^i)^H \theb_t \|^2} 
\end{align}
Note that direct links are much weaker than RIS links, i.e., $h_{UB,t}^i \ll (\g_t^i)^H \theb_t$, and we have $\| \theb_t\|=\sqrt{N}$. Thus,
\begin{equation}
    0 < \| (\g_t^i)^H \theb_t\|^2 \leq N^2 g_m^2.
\end{equation}
\begin{align}
    &\Eb \| \Delta_{UB,t} + [\h_{UR,t}^i diag(\boldsymbol{\Delta}_{RB,t})+(\boldsymbol{\Delta}_{UR,t}^i)^H diag(\h_{RB,t}) \nonumber \\
    & + (\boldsymbol{\Delta}_{UR,t}^i)^H diag(\boldsymbol{\Delta}_{RB,t})]\theb_t\|^2 \\
    &= \Eb [\|\Delta_{UB,t} \|^2 + 2 \Delta_{UB,t} [\h_{UR,t}^i diag(\boldsymbol{\Delta}_{RB,t}) \nonumber \\
    &+(\boldsymbol{\Delta}_{UR,t}^i)^H diag(\h_{RB,t}) + (\boldsymbol{\Delta}_{UR,t}^i)^H diag(\boldsymbol{\Delta}_{RB,t})]\theb_t ] \nonumber \\
    &+ \Eb \| [\h_{UR,t}^i diag(\boldsymbol{\Delta}_{RB,t})+(\boldsymbol{\Delta}_{UR,t}^i)^H diag(\h_{RB,t}) \nonumber \\
    &+ (\boldsymbol{\Delta}_{UR,t}^i)^H diag(\boldsymbol{\Delta}_{RB,t})]\theb_t\|^2 \\
    & \leq \sigt_h^2 + \Eb \| \h_{UR,t}^i diag(\boldsymbol{\Delta}_{RB,t})+(\boldsymbol{\Delta}_{UR,t}^i)^H diag(\h_{RB,t}) \nonumber \\
    & + (\boldsymbol{\Delta}_{UR,t}^i)^H diag(\boldsymbol{\Delta}_{RB,t})\|^2 \|\theb_t\|^2 \\
    & \leq \sigt_h^2 + N (\Eb\| (\h_{UR,t}^i)^H diag (\boldsymbol{\Delta}_{RB,t})\|^2 + \nonumber \\
    & \Eb \|(\boldsymbol{\Delta}_{UR,t})^H diag(\h_{RB,t}) \|^2 + \Eb \|(\boldsymbol{\Delta}_{UR,t}^i)^H diag(\boldsymbol{\Delta}_{RB,t}) \|^2 ) \nonumber \\
    & \leq \sigt_h^2 + N^2 \sigt_h^2 (h_{UR,a}^2 + h_{RB,a}^2 + \sigt_h^2).
\end{align}
Hence,
\begin{equation}
    \Eb \|A \|^2 \leq \frac{\sigt_h^2 (1 + N^2(h_{UR,a}^2 + h_{RB,a}^2 + \sigt_h^2))}{\|h_{UB,m} \|^2}.
\end{equation}
Define $C = \frac{\sigt_h^2 (1 + N^2(h_{UR,a}^2 + h_{RB,a}^2 + \sigt_h^2))}{\|h_{UB,m} \|^2}$, we have:
\begin{equation}
    \Eb \bigg\| \frac{h_t^{u,i}}{\hh_t^{u,i}} \bigg\|^2 \leq 1 + C, \quad \Eb \bigg\| 1 - \frac{h_t^{u,i}}{\hh_t^{u,i}} \bigg\|^2 = \Eb \|A\|^2 \leq C.
\end{equation}
By plugging back the above inequalities to error terms, we get Corollary~\ref{cor:convergence}.

\vspace{-0.3in}
\section{Proof of Theorem 2}
\label{sec:profpersonal}
We first show that the local personalized objective function $R_i$ is $\sqrt{2L^2 + 2 \lambda^2}$-smooth and has bounded stochastic gradient with $\sigma^2$. Define $L_2 =\sqrt{2L^2 + 2 \lambda^2}$.
\begin{align}
    &\|\nabla R_i(\v_1;\w) - \nabla R_i(\v_2;\w) \|^2 = \|\nabla F_i(\v_1)+\lambda(\v_1 - \w) \nonumber \\
    &- (\nabla F_i(\v_2)+\lambda(\v_2 - \w))\|^2 = \|\nabla F_i(\v_1) - \nabla F_i(\v_2) \nonumber \\
    & + \lambda (\v_1 - \v_2)\|^2 \leq 2 \|\nabla F_i(\v_1) - \nabla F_i(\v_2) \|^2 + 2 \lambda^2 \|\v_1 - \v_2 \|^2 \nonumber \\
    & \leq (2L^2 + 2 \lambda^2) \|\v_1 - \v_2 \|^2.
\end{align}
\begin{align}
    &\Eb \|g_i(\v;\w) - \nabla R_i(\v;\w) \|^2 = \Eb \|\nabla F_i(\v,\xi_i) + \lambda(\v - \w)  \nonumber \\
    &- \nabla F_i(\v) - \lambda(\v - \w)\|^2 \nonumber \\
    &= \Eb \|\nabla F_i(\v,\xi_i) - \nabla F_i(\v)\|^2 \leq \sigma^2.
\end{align}
We take conditional expectation with respect to $\v^i_t$ and use smoothness to compute one step descent:
\begin{align}
    &\Eb_t [R_i(\v^i_{t+1};\w_s)] - R_i (\v^i_{t};\w_s) \leq \nonumber \\
    &<\nabla R_i(\v^i_t;\w_s), \Eb_t [\v^i_{t+1} - \v^i_{t}]> + \frac{L_2}{2} \Eb_t \|\v^i_{t+1} - \v^i_{t}\|^2 \nonumber \\
    & = - <\nabla R_i(\v^i_t;\w_s), \Eb_t [\eta_v g_i(\v^i_t;\w_s)]> \nonumber \\
    &+ \frac{L_2}{2} \eta_v^2 \Eb_t \|g_i(\v^i_t;\w_s)\|^2 = - \eta_v \|\nabla R_i(\v^i_t;\w_s) \|^2 \nonumber \\
    & + \frac{L_2}{2} \eta_v^2  \Eb_t \|g_i(\v^i_t;\w_s) - \nabla R_i(\v^i_t;\w_s) + \nabla R_i(\v^i_t;\w_s) \|^2 \nonumber \\
    & = - \eta_v \|\nabla R_i(\v^i_t;\w_s) \|^2 + \frac{L_2}{2} \eta_v^2 \Eb_t \|g_i(\v^i_t;\w_s) - \nabla R_i(\v^i_t;\w_s) \|^2 \nonumber \\
    & + \frac{L_2}{2} \eta_v^2 \Eb_t \|\nabla R_i(\v^i_t;\w_s) \|^2 \nonumber \\
    & = - (\eta_v - \frac{L_2}{2} \eta_v^2) \|\nabla R_i(\v^i_t;\w_s) \|^2 + \frac{L_2}{2} \eta_v^2 \sigma^2
\end{align}
Rearrange and average over $\tau_v^i$ steps,
\begin{align}
    \frac{1}{\tau_v^i} \sum_{t=0}^{\tau_v^i - 1} \Eb \|\nabla R_i(\v_t^i;\w_s) \|^2 &\leq \frac{R_i(\v_0^i;\w_s) - R_i(\v_{\tau_v-1}^i;\w_s)}{\tau_v (\eta_v - \frac{L_2}{2} \eta_v^2)} \nonumber \\
    & + \frac{L_2}{2} \eta_v^2 \sigma^2.
\end{align}
\begin{align}
    &\Eb \|\nabla R_i(\v_t^i;\w^*) \|^2 = \Eb \|\nabla R_i(\v_t^i;\w^*) - \nabla R_i(\v_t^i;\w_s)  \nonumber \\
    &+ \nabla R_i(\v_t^i;\w_s)\|^2 = \Eb \| \lambda(\w_s - \w^*) + \nabla R_i(\v_t^i;\w_s)\|^2 \nonumber \\
    & \leq 2 \lambda^2 \Eb \|\w_s - \w^*\|^2 + 2 \Eb \| \nabla R_i(\v_t^i;\w_s)\|^2
\end{align}
Next, we find the upper bound for $\Eb \|\w_s - \w^*\|^2$.
\begin{align}
    &\w_{T} - \w_{s} = - \sum_{l=s}^{T} \bigg( \sum_{i=1}^{m} \frac{\alpha_i}{\tau^i_l}\frac{h_l^{u,i}}{\hh_l^{u,i}} \eta_l \sum_{k=0}^{\tau^i_l-1} \left(\nabla F_i(\w^i_{l, k}, \xi^i_{l, k})\right) + \tilde{\z}_l^u \bigg) \nonumber
\end{align}
\begin{align}
    &\Eb \|\w_{T} - \w_{s}\|^2 \leq \Eb \| \sum_{l=s}^{T} \tilde{\z}_l^u\|^2 \nonumber \\
    & +\Eb \| \sum_{l=s}^{T} \sum_{i=1}^{m} \frac{\alpha_i}{\tau^i_l}\frac{h_l^{u,i}}{\hh_l^{u,i}} \eta_l \sum_{k=0}^{\tau^i_l-1} \left(\nabla F_i(\w^i_{l, k}, \xi^i_{l, k})\right)\|^2 \nonumber \\
    & \leq (T-s) \sum_{l=s}^{T} \Eb \|\sum_{i=1}^{m} \frac{\alpha_i}{\tau^i_l}\frac{h_l^{u,i}}{\hh_l^{u,i}} \eta_l \sum_{k=0}^{\tau^i_l-1} \left(\nabla F_i(\w^i_{l, k}, \xi^i_{l, k})\right)\|^2 \nonumber \\
    & + (T-s) \sum_{l=s}^{T} \bigg(\frac{\sigma_{c,u}}{\beta_l^u} \bigg)^2 \nonumber \\
    & \leq (T-s) \sum_{l=s}^{T} m \sum_{i=1}^{m} \Eb \| \frac{\alpha_i}{\tau^i_l}\frac{h_l^{u,i}}{\hh_l^{u,i}} \eta_l \sum_{k=0}^{\tau^i_l-1} \nabla F_i(\w^i_{l, k}, \xi^i_{l, k})\|^2 \nonumber \\
    & + (T-s) \sum_{l=s}^{T} \bigg(\frac{\sigma_{c,u}}{\beta_l^u} \bigg)^2 \nonumber \\ 
    & \leq (T-s) \sum_{l=s}^{T} m \sum_{i=1}^{m} \eta_l^2 \frac{\alpha_i^2}{\tau^i_l} \Eb \bigg\|\frac{h_l^{u,i}}{\hh_l^{u,i}} \bigg\|^2 \sum_{k=0}^{\tau^i_l-1} \Eb \|\nabla F_i(\w^i_{l, k}, \xi^i_{l, k})\|^2 \nonumber \\
    & + (T-s) \sum_{l=s}^{T} \bigg(\frac{\sigma_{c,u}}{\beta_l^u} \bigg)^2 \nonumber \\
    & = (T-s) \sum_{l=s}^{T} \bigg[ m \eta_l^2 G^2 \sum_{i=1}^{m} \alpha_i^2 \Eb \bigg\|\frac{h_l^{u,i}}{\hh_l^{u,i}} \bigg\|^2 + \bigg(\frac{\sigma_{c,u}}{\beta_l^u} \bigg)^2 \bigg]
\end{align}
Now rearrange and average over $T$ global rounds and $\tau_v^i = \tau_v$ local steps, use $\eta_l = \eta$, $\frac{1}{\bar{\beta}^2} = \frac{1}{T} \sum_{t=0}^{T-1} \frac{1}{(\beta_t^u)^2}$,
\begin{align}
    &\frac{1}{T}\sum_{s=0}^{T-1} \frac{1}{\tau_v} \sum_{t=0}^{\tau_v-1} \Eb \|\nabla R_i(\v^i_t;\w^*) \|^2 \leq \nonumber \\
    & \frac{1}{T}\sum_{s=0}^{T-1} \frac{1}{\tau_v} \sum_{t=0}^{\tau_v-1} 2 \lambda^2 \Eb \|\w_s - \w^*\|^2 + 2 \Eb \| \nabla R_i(\v_t^i;\w_s)\|^2 \nonumber \\
    & = \frac{1}{T}\sum_{s=0}^{T-1} 2 \lambda^2 \bigg[ (T-s) \sum_{l=s}^{T} m \eta^2 G^2 \sum_{i=1}^{m} \alpha_i^2 \Eb \bigg\|\frac{h_l^{u,i}}{\hh_l^{u,i}} \bigg\|^2 \nonumber \\
    &+ (T-s)^2 \bigg(\frac{\sigma_{c,u}}{\beta} \bigg)^2 \bigg] + \sqrt{2L^2 + 2 \lambda^2} \eta_v^2 \sigma^2 \nonumber \\
    &+ \frac{1}{T}\sum_{s=0}^{T-1} \frac{2[R_i(\v_s^i;\w_s) - R_i(\v_{s+1}^i;\w_s)]}{\tau_v (\eta_v - \frac{L_2}{2}\eta_v^2)}
\end{align}
where $\v_{s}^i$ is the last local step in global round $s$, i.e, $\v_{s}^i=\v_{s,\tau_v}^i$. 
Note that if $T\geq 4$, which is practical, then
\begin{align}
    \frac{2}{T} \sum_{s=0}^{T-1} (T-s)^2 = \frac{2}{T} \times \frac{T (T+1) (2T+1)}{6} \leq T^2
\end{align}
We arrive the final result:
\begin{align}
    & \min \Eb \|\nabla R_i(\v^i_t;\w^*) \|^2 \leq \sqrt{2L^2 + 2 \lambda^2} \eta_v^2 \sigma^2 \nonumber \\
    & \frac{1}{T}\sum_{s=0}^{T-1} 2 \lambda^2m \eta^2 G^2 (T-s) \sum_{l=s}^{T} \sum_{i=1}^{m} \alpha_i^2 \Eb \bigg\|\frac{h_l^{u,i}}{\hh_l^{u,i}} \bigg\|^2 + \lambda^2T^2 \frac{\sigma_{c,u}^2}{\beta^2} \nonumber \\
    & + \frac{1}{T}\sum_{s=0}^{T-1} \frac{2[R_i(\v_s^i;\w_s) - R_i(\v_{s+1}^i;\w_s)]}{\tau_v (\eta_v - \frac{\sqrt{2L^2 + 2 \lambda^2}}{2}\eta_v^2)}.
\end{align}